\documentclass{article}

\usepackage{arxiv}

\usepackage[utf8]{inputenc} 
\usepackage[T1]{fontenc}    
\usepackage{hyperref}       
\usepackage{url}            
\usepackage{booktabs}       
\usepackage{amsfonts}       
\usepackage{nicefrac}       
\usepackage{microtype}      
\usepackage{lipsum}
\usepackage{graphicx}
\usepackage{authblk}
\usepackage{enumitem}
\usepackage{comment}
\usepackage{subfigure}
\usepackage[percent]{overpic}
\usepackage{bm}
\usepackage{amssymb} 
\usepackage{amsmath,amsthm}
\usepackage{mathrsfs}
\usepackage{arydshln}
\newtheorem{proposition}{Proposition}

\graphicspath{ {./images/} }
\title{A Data-driven Convergence Booster for Accelerating and Stabilizing Pseudo Time-stepping}

\author[1]{Xukun Wang}
\author[2]{Yilang Liu}
\author[3]{Xiang Yang}
\author[2]{Weiwei Zhang \thanks{Corresponding author: aeroelastic@nwpu.edu.cn}}

\affil[1]{School of Aeronautics, Universidad Politécnica de Madrid, Madrid, 28040, Spain}
\affil[2]{National Key Laboratory Science and Technology on Aerodynamic Design and Research, School of Aeronautics, Northwestern Polytechnical University, Xi'an, 710072, China}
\affil[3]{Department of Mechanical Engineering, Pennsylvania State University, University Park, 16802, PA, USA}

\begin{document}

\maketitle
\begin{abstract}
This paper introduces a novel data-driven convergence booster that not only accelerates convergence but also stabilizes solutions in cases where obtaining a steady-state solution is otherwise challenging. The method constructs a reduced-order model (ROM) of the solution residual using intermediate solutions and periodically solves a least-square problem in the low-dimensional ROM subspace. The second-order approximation of the residual and the use of normal equation distinguish this work from similar approaches in the literature from the methodology perspective. From the application perspective, in contrast to prior studies that focus on linear systems or idealized problems, we rigorously assess the method’s performance on realistic computational fluid dynamics (CFD) applications. In addition to reducing the time complexity of point-iterative solvers for linear systems, we demonstrate substantial reductions in the number of pseudo-time steps required for implicit schemes solving the nonlinear Navier–Stokes equations. Across a range of two- and three-dimensional flows—including subsonic inviscid and transonic turbulent cases—the method consistently achieves a 3 to 4 times speedup in wall-clock time.
Lastly, the proposed method acts as a robust stabilizer, capable of converging to steady solutions in flows that would otherwise exhibit persistent unsteadiness—such as vortex shedding or transonic buffet—without relying on symmetry boundary conditions.
\end{abstract}


\section{Introduction}

Computational Fluid Dynamics (CFD) is an essential tool for simulating fluid flow in engineering applications, including aviation, naval vessels, automobiles, and wind turbines, as well as in scientific research on turbulence, shock waves, and biofluid dynamics~\cite{mani2023perspective}. However, large-scale, high-resolution simulations pose significant challenges: achieving higher accuracy requires finer resolutions, while the associated computational costs remain prohibitive~\cite{yang2021grid,choi2012grid}. To address this, acceleration methods play a crucial role in CFD by reducing computational expense, expediting convergence, and, in some cases, enhancing numerical stability. These methods can generally be categorized into classical numerical techniques and emerging data-driven approaches.

Classical acceleration methods refer to canonical numerical techniques that form the foundation of CFD solvers. One widely used approach is local time-stepping~\cite{RN153}, which adjusts the time step based on local flow conditions to accelerate convergence to solutions. However, this technique is constrained by the Courant–Friedrichs–Lewy (CFL) condition~\cite{5391985}, which imposes a limit on the maximum allowable time step. To overcome this limitation, implicit time-marching schemes have been developed, often combined with iterative solvers such as the Lower Upper Symmetric Gauss-Seidel (LU-SGS) method~\cite{RN10} and the Generalized Minimal Residual (GMRES) method~\cite{RN261,RN593} to further enhance convergence. Additionally, residual smoothing techniques~\cite{RN594} mitigate high-frequency errors by averaging residuals across adjacent grid elements.
Another fundamental acceleration strategy is the multigrid method, originally developed for elliptic equations~\cite{RN627} and later extensively applied in CFD simulations~\cite{RN163}. This method alternates between smoothing high-frequency errors using time-marching schemes and eliminating low-frequency errors via coarse-grid corrections. Beyond $h$-multigrid, which solves governing equations on multiple mesh resolutions, $p$-multigrid methods~\cite{RN1112,RN1050} employ different polynomial orders instead of grid coarsening, achieving significant acceleration without modifying the mesh. Other classical acceleration techniques include preconditioning strategies~\cite{RN991} and enthalpy damping~\cite{RN1180}, both of which improve numerical efficiency and stability.

Despite their success, classical methods struggle with highly complex or nonlinear problems. Furthermore, large-scale, high-resolution CFD simulations remain computationally expensive even when multiple acceleration techniques are applied simultaneously, motivating continued research on convergence boosters. The advent of machine learning (ML), artificial intelligence (AI), and big data has led to the emergence of new, data-driven approaches in fluid dynamics~\cite{RN75}, leveraging intermediate simulation data and ML techniques to improve convergence and efficiency.

The use of data for accelerating numerical solvers dates back to classical vector extrapolation methods (VEM), such as minimal polynomial extrapolation (MPE) and reduced rank extrapolation (RRE)~\cite{MPE,RN264,MESINA1977165}. These methods exploit the exponential decay of numerical errors, approximating the convergent solution as a linear combination of previous snapshots. Their theoretical foundations and practical implementations have been extensively reviewed in the literature~\cite{RN128,RN170,RN139} and is not repeated here for brevity. In the 1980s, VEM techniques were integrated into CFD solvers to accelerate convergence~\cite{doi:10.2514/6.1987-1143,doi:10.2514/6.1990-338,RN137}. 
More recent developments build on these concepts but leverage reduced-order models (ROMs), e.g.,  proper orthogonal decomposition (POD)~\cite{RN600} and dynamic mode decomposition (DMD)~\cite{RN5}, that are originally developed for flow analysis and dimension reduction. Building on this idea, Liu et al.~\cite{RN1} proposed the mode multigrid method (MMG), which employs DMD to accelerate convergence by filtering high-order error modes. This approach has since been extended to solving turbulent flows~\cite{RN40} and adjoint solvers~\cite{RN602}. More recently, Bin et al.~\cite{BIN2025113859} projected intermediate solutions onto a low-dimensional Hilbert subspace and directly solved the discretized system, achieving reduction of the time complexity of baseline point iterative methods.

Another class of data-driven acceleration methods leverages advances in ML techniques, where data are used for offline training. Deep neural networks (DNNs), for example, have demonstrated significant potential in accelerating iterative solvers, improving initial guesses, and replacing components of the iterative process~\cite{RN1343,RN66}. Reinforcement learning (RL) has also been employed to dynamically optimize solver parameters, enhancing convergence efficiency~\cite{RN1001}. More advanced neural network frameworks, such as graph convolutional networks~\cite{RN1274} and probabilistic generative diffusion models~\cite{NING2025109917}, were also reported to accelerate convergence. A comprehensive review of these ML-based approaches is beyond the scope of this paper, but recent surveys provide in-depth discussions on their implementation and impact on CFD acceleration~\cite{RN72}.

Among acceleration techniques, some not only enhance computational speed but also address convergence and numerical stability issues. These methods, termed \emph{convergence boosters}, serve a dual role: accelerating the solver when iterations are convergent and stabilizing it when convergence is slow or when divergence occurs. For example, the BoostConv method proposed by Citro et al.~\cite{RN118} recombines residuals from previous steps to either accelerate or stabilize iterations. Similarly, Cao et al.~\cite{RN640} introduced an optimization-enhanced ROM framework that improves steady-state solver convergence. More recently, Wang et al.~\cite{RN489} extended the MMG method to improve both convergence and stability in CFD solvers, demonstrating its effectiveness as a convergence booster.

In this study, we propose a data-driven convergence booster, named Mean-based Minimal Residual (MMRES), which enhances both acceleration and stability in iterative solvers. MMRES constructs a mean-based ROM from solution snapshots and computes the optimal solution within this low-dimensional subspace by minimizing the residual norm. The method is rigorously analyzed in relation to existing acceleration techniques, including RRE, Anderson acceleration, and quasi-Newton iteration methods. Furthermore, we demonstrate that MMRES significantly reduces the time complexity of baseline methods, such as the Jacobi iteration, from $O(n^2)$ to $O(n)$ when solving linear problems, making it particularly advantageous for large-scale simulations. Additionally, MMRES effectively accelerates and stabilizes implicit pseudo-time marching schemes across a wide range of cases, showcasing its potential for tackling complex nonlinear problems. Importantly, our goal is to enhance existing CFD solvers through MMRES rather than replace them entirely.

The remainder of this paper is organized as follows. Section~\ref{sec:methodology} presents the methodology, detailing the iterative formulation and the proposed MMRES approach. Section~\ref{sec:complexity} analyzes the computational complexity of MMRES and verifies its performance in linear test cases. Section~\ref{sec:cfd} demonstrates its effectiveness in various fluid dynamics problems, highlighting both acceleration and stabilization benefits. Finally, Section~\ref{sec:conclusion} summarizes the findings and outlines potential future research directions.

\section{Methodology}
\label{sec:methodology}

Consider a discretized nonlinear problem:
\begin{equation}
	\label{eq:nonlinearsystem}
	\mathcal{N}(x) = 0.
\end{equation}
A general iterative update can be written as:
\begin{equation}
	\label{eq:generaliteration}
	x_{k+1} = x_k + B_k r_k,
\end{equation}
where $x_k \in \mathbb{R}^n$ is the solution vector, $r_k \in \mathbb{R}^n$ is the residual, and $B_k \in \mathbb{R}^{n \times n}$ depends on the iteration scheme. Without loss of generality, the number of degrees of freedom (DOFs) is given by $n = n_i^d$, where $n_i$ is the number of DOFs per dimension and $d$ is the spatial dimension. In a linear system, $B_k$ remains constant, whereas in most flow problems, it depends on the pseudo-time marching scheme.

The core idea of the proposed method is to accelerate convergence by periodically solving the problem in a reduced-order subspace and using the obtained solution as an initial guess for subsequent iterations. To achieve this, we construct a reduced-order model (ROM) based on intermediate solution snapshots. Given $m$ stored solution snapshots $\{x_1, x_2, \dots, x_m\}$, we define a mean-based ROM:
\begin{equation}
	\label{eq:ROM}
	\tilde{x} = \bar{x} + \Phi\xi,
\end{equation}
where the mean solution is:
\begin{equation}
	\bar{x} = \frac{1}{m} \sum_{i = 1}^{m} x_i.
\end{equation}
The basis matrix $\Phi \in \mathbb{R}^{n \times m}$ is constructed from the deviations of snapshots from their mean:
\begin{equation}
	\label{eq:Phi}
	\Phi =
	\begin{bmatrix}
		| & | &  & | \\
		\phi_1 & \phi_2 & \cdots & \phi_m \\
		| & | &  & | \\
	\end{bmatrix}
	=
	\begin{bmatrix}
		| & | &  & | \\
		x_1-\bar{x} & x_2-\bar{x} & \cdots & x_m-\bar{x} \\
		| & | &  & | \\
	\end{bmatrix}.
\end{equation}
In Eq.~\eqref{eq:ROM}, $\xi \in \mathbb{R}^m$ is a coefficient vector that determines the optimal correction within the reduced-order space.

The objective is to find $\xi^*$ that minimizes the residual:
\begin{equation}
	\xi^{\ast} = \mathop{\arg\min}\limits_{\xi \in \mathbb{R}^m} \Vert r(\tilde{x}) \Vert_2,
    \label{eq:LSprob}
\end{equation}
where $r(\tilde{x})$ is the residual of the projected solution. Instead of solving for $\xi^*$ directly, we approximate the residual function using a first-order Taylor expansion around the converged solution $x^*$:
\begin{equation}
	r(x) \approx r(x^{\ast}) + J^{\ast} (x - x^{\ast}) + O\left(\|x-x^{\ast}\|^2\right),
\end{equation}
where $J^{\ast} = \frac{\partial r}{\partial x}(x^{\ast})$ is the Jacobian matrix at convergence. Neglecting higher-order terms and using $r(x^*) = 0$, we obtain:
\begin{equation}
	\label{eq:linear_residual}
	r(x) \approx J^{\ast} x + b,
\end{equation}
where $b = -J^{\ast} x^{\ast}$ is a constant vector. Substituting Eq.~\eqref{eq:ROM} into Eq.~\eqref{eq:linear_residual}, we obtain:
\begin{equation}
	\label{eq:linear_residual2}
	r(\tilde{x}) \approx r(\bar{x}) + J^{\ast} \Phi \xi,
\end{equation}
where $r(\bar{x})$ is the residual of the mean solution.
The quadratic approximation of the residual distinguishes the present method from those in, e.g., \cite{RN20}, where a linear system is solved to minimize the $J^{\ast}$-norm of error. As will become clear, despite searching the solution in the same reduced-order space, a better solution can be achieved by MMRES, leading to more significant/robust acceleration than those reported in the literature.

Minimizing $\Vert r(\tilde{x}) \Vert_2$ reduces to solving a least-squares (LS) problem:
\begin{equation}
	\label{eq:LS}
	\Psi \xi = -\bar{r},
\end{equation}
where $\Psi = J^{\ast} \Phi$, and according to Equations~\eqref{eq:Phi} and~\eqref{eq:linear_residual}, it is:
\begin{equation}
	\label{eq:Psi}
	\Psi =
	\begin{bmatrix}
		| & | &  & | \\
		\psi_1 & \psi_2 & \cdots & \psi_m \\
		| & | &  & | \\
	\end{bmatrix}
	=
	\begin{bmatrix}
		| & | &  & | \\
		r(x_1)-\bar{r} & r(x_2)-\bar{r} & \cdots & r(x_m)-\bar{r} \\
		| & | &  & | \\
	\end{bmatrix}.
\end{equation}
The solution to Eq.~\eqref{eq:LS} is obtained using the Moore-Penrose inverse:
\begin{equation}
	\label{eq:optimal_xi}
	\xi^{\ast} = -\Psi^{+} \bar{r},
\end{equation}
where $\Psi^{+} = (\Psi^{\rm T} \Psi)^{-1} \Psi^{\rm T}$. Finally, the optimal projected solution is:
\begin{equation}
	\label{eq:optimal soultion}
	\tilde{x}^{\ast} = \bar{x} - \Phi (\Psi^{\rm T} \Psi)^{-1} \Psi^{\rm T} \bar{r}.
\end{equation}
If $\tilde{x}^{\ast}$ provides a good approximation of the convergent solution ($\tilde{x}^{\ast} \approx x^{\ast}$), using it as a new initial guess enhances the solver’s convergence. Since this ROM is based on the mean and minimizes the residual under the linearity assumption in Eq.~\eqref{eq:linear_residual}, we call this method Mean-based Minimal Residual (MMRES). As for the comparison with methods in previous literature, three remarks are given in the Appendix\ref{sec:AppendixB}, where its relation with RRE, AA and quasi-Newton iteration are fully discussed.



\section{Time Complexity Analysis for Linear Systems}
\label{sec:complexity}

In this section, we apply MMRES to accelerate the Jacobi method to solve linear equations and analyze its impact on time complexity. The time complexity of an iterative method is defined in terms of the number of iterations, $N_{\epsilon}$, required to reduce the residual by a factor of $1/\epsilon$. Specifically, we express its scaling behavior with the number of degrees of freedom (DOFs) per dimension, $n_i$, as:
\begin{equation}
    N_{\epsilon} \sim n_i^{\alpha}.
\end{equation}
The exponent $\alpha$ characterizes the time complexity. A related metric is the convergence rate, $R$, defined as:
\begin{equation}
    \label{eq:1R}
    R = \left|\frac{{\rm log}(\epsilon)}{N_{\epsilon}}\right|.
\end{equation}

We first analyze the time complexity of the baseline Jacobi method. Consider the one-dimensional Poisson equation with periodic boundary conditions, discretized using a three-point finite difference scheme to form a linear system $Ax = b$. Without loss of generality, we assume the computational domain is $[0,1]$ with $n_i+1$ uniformly spaced points. Due to periodic boundary conditions, the number of independent DOFs is $n_i$. The Jacobi iteration for this system is given by:
\begin{equation}
    \label{eq:Jacobiiteration}
    x_{k+1} = G_{J}x_k + f,
\end{equation}
where $G_J$ is the Jacobi iteration matrix, defined as $G_J = I - D^{-1}A$, with $D$ being the diagonal of $A$. The asymptotic convergence factor of the Jacobi iteration is determined by the spectral radius of $G_J$:
\begin{equation}
    \phi = \lim_{k \to \infty} \left( \Vert G_J^k \Vert \right)^{1/k} = \rho(G_J),
\end{equation}
where $\rho(G_J)$ denotes the spectral radius and is the amount the residual reduces in one iteration. For the one-dimensional Poisson equation discretized using a three-point difference scheme, the spectral radius is:
\begin{equation}
    \rho(G_J) = 1 - 2\sin^2\left(\frac{\pi}{2n_i}\right) \approx 1 - \frac{{\pi}^2}{2n^2_i}.
\end{equation}
Using Eq.~\eqref{eq:1R}, the number of iterations required to reach a residual reduction of $\epsilon$ is:
\begin{equation}
    N_{\epsilon} \sim \frac{1}{R} \approx \frac{1}{{\rm log}(\rho(G_J))} \approx - \frac{1}{{\pi}^2/2n_i^2} \sim O(n_i^2).
\end{equation}
Thus, the Jacobi method exhibits a time complexity of $O(n_i^2)$.

Next, we analyze the time complexity of the MMRES-accelerated Jacobi method. MMRES modifies the iteration process by searching for an improved solution within a reduced subspace. Specifically, it finds $\tilde{x} \in x_1 + \mathcal{K}_{m}(G_J, \Delta x_1)$ such that $\tilde{r} \perp \mathcal{L}$, where
\begin{equation}
    \mathcal{K}_{m}(G_J, \Delta x_1) = {\rm span}\{\Delta x_1, G_J\Delta x_1, \dots, G_J^{m-1}\Delta x_1\}
\end{equation}
is the Krylov subspace, $\mathcal{L} = A\mathcal{K}_{m}(G_J, \Delta x_1)$, and $\tilde{r} = b - A\tilde{x}$ is the residual. Since this formulation is analogous to GMRES, the convergence properties of MMRES can be estimated using established results from GMRES. In the following, we outline this estimate.

Following~\cite{RN1303}, the residual norm at the $m^{th}$ step of GMRES satisfies:
\begin{equation}
    \label{eq:GMRES_convergnce_estimate}
    \Vert r_m \Vert \leq \kappa(W) \varepsilon^{(m)} \Vert r_0 \Vert,
\end{equation}
where $\kappa(W) = \Vert W \Vert \Vert W^{-1} \Vert$ is the condition number of the eigenvector matrix $W$, and $\varepsilon^{(m)}$ is the polynomial approximation error:
\begin{equation}
    \varepsilon^{(m)} = \min_{p \in P_m, p(0) = 1 } \max_{\lambda_i} \vert p(\lambda_i)\vert
\end{equation}
where $P_m$ is the polynomial space up to $m$-order.
For matrices with clustered eigenvalues, an explicit bound on $\varepsilon^{(m)}$ is given by:
\begin{equation}
    \label{eq: epsilon^m}
    \varepsilon^{(m)} \leq \left[\frac{D}{d}\right]^{\nu}\left[\frac{R}{C}\right]^{m - \nu},
\end{equation}
where $D = \max_{j = 1,\nu; k =\nu +1,n_i}\vert \lambda_j - \lambda_k \vert$, $d = \min_{j=1,\nu}$, $C$ is the center, and $R$ is the radius of the enclosing circle for the eigenvalues.

For the Jacobi iteration matrix $G_J$, the eigenvalues are:
\begin{equation}
    \lambda_j(G_J) = 1 - 2 \sin^2\left(\frac{j\pi}{2(n+1)}\right) = \cos\left(\frac{j\pi}{n+1}\right),\ {\rm for}\ j = 1,2,\dots, n_i.
\end{equation}
Without loss of generality, we assume $n_i$ is even and we have:
\begin{align}
    \label{eq: niu}
    \nu &= n/2 \\
    \label{eq：DdCR}
    D = 2c_1,\ d = c_2, C &= \frac{1}{2}\left(c_1+c_2\right)\ \text{and}~\ R = \frac{1}{2}\left(c_1-c_2\right).
\end{align}
where $c_1 = \cos\left(\frac{\pi}{n_i+1}\right)$ and $c_2 = \cos\left(\frac{n\pi}{2(n_i+1)}\right)$.

Substituting Eq.~\eqref{eq: epsilon^m} and ~\eqref{eq: niu} into Eq.~\eqref{eq:GMRES_convergnce_estimate} and moving $\Vert r_0\Vert$ to the left-hand side of inequality, we will get:
\begin{equation}
    \frac{\Vert r_m \Vert}{\Vert r_0 \Vert} \leq \kappa(W)\varepsilon^{(m)} \leq \kappa(W)\left[\frac{D}{d}\right]^{n/2}\left[\frac{R}{C}\right]^{m - n/2}
\end{equation}

Taking $m = N_{\epsilon}$, to meet the requirement of iteration(residual is reduced by a factor of $1/\epsilon$), the following inequality should be satisfied:
\begin{equation}
	\frac{\Vert \tilde{r} \Vert}{\Vert r_0 \Vert} \leq \kappa(W)\varepsilon^{(N_\epsilon)} \leq \kappa(W)\left[\frac{D}{d}\right]^{n_i/2}\left[\frac{R}{C}\right]^{N_\epsilon - n_i/2} \leq \epsilon.
\end{equation}

Since $G_J$ is normal, its eigenvectors form an orthonormal basis, and thus, its condition number satisfies $\kappa_2(W) = \Vert W \Vert_2 \Vert W^{-1} \Vert_2 = 1$, which leads to:
\begin{equation}
    \label{eq: DdRClesse}
    \left[\frac{D}{d}\right]^{n_i/2}\left[\frac{R}{C}\right]^{N_\epsilon - n_i/2} \leq \epsilon.
\end{equation}
Substituting Eq.~\eqref{eq：DdCR} into Eq.~\eqref{eq: DdRClesse}, moving the terms associated with $N_\epsilon$ and $n_i$ to the both sides of the inequality and taking the logarithm, we will get:
\begin{equation}
	N_\epsilon \log \left(\frac{c_1 - c_2}{c_1 + c_2}\right) \leq \frac{n_i}{2} \left[\log\left(\frac{c_1 - c_2}{c_1+c_2}\right) + \log\left(\frac{c_2}{2c_1}\right)\right] + \log \epsilon \leq  \frac{n_i}{2}\log \left(\frac{c_1 - c_2}{c_1 + c_2}\right) + \log \epsilon
\end{equation}
and considering that:
\begin{equation}
	\log \left(\frac{c_1 - c_2}{c_1 + c_2}\right) = \log \left(1 - \frac{2c_2}{c_1 + c_2}\right) \sim \frac{1}{n_i},
\end{equation}
we finally get:
\begin{equation}
	N_\epsilon \sim \frac{\log \epsilon}{1/n_i}+\frac{n_i}{2} \sim O(n_i).
\end{equation}
This confirms that MMRES reduces the time complexity of the Jacobi iteration from $O(n_i^2)$ to $O(n_i)$.

To validate this analysis, we apply MMRES to the Jacobi method to solve the Poisson equation in one- and two-dimensional cases, where the forcing terms are arbitrarily given and are $f(x) = \sin(\pi x)$ and $f(x,y) = \sin(\pi x)\sin(\pi y)$, respectively. The convergence rate is measured using the inverse of the convergence factor, $1/R$, as defined in Eq.~\eqref{eq:1R}. Figure~\ref{fig:scaling} presents the results, showing that the baseline Jacobi method scales as $1/R \sim n_i^2$, while MMRES scales as $1/R \sim n_i$.  
Following the discussion in \cite{BIN2025113859}, this reduction in time complexity extends to other point-iterative solvers and three-dimensional problems. However, as this work focuses on the practical benefits of the proposed methodology, a more detailed analysis on these fundamental problems is omitted for brevity.


\begin{figure}[!t]
    \centering
    \subfigure {\
        \begin{minipage}[b]{.46\linewidth}
            \centering
            \begin{overpic}[scale=0.85]{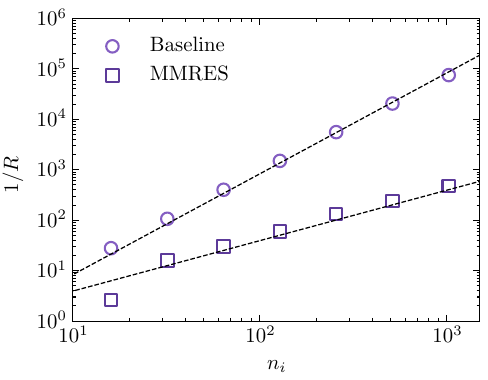}
                \put(0,75){(a)}
            \end{overpic}
        \end{minipage}
    }
    \subfigure {\
        \begin{minipage}[b]{.46\linewidth}
            \centering
            \begin{overpic}[scale=0.85]{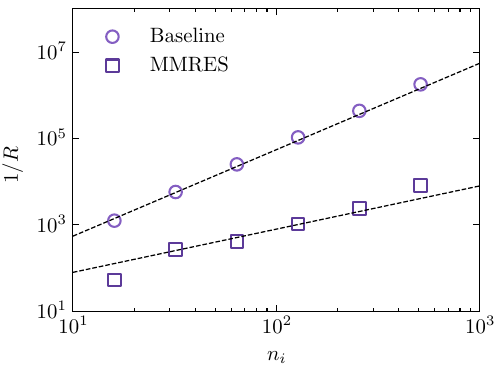}
                \put(0,75){(b)}
            \end{overpic}
        \end{minipage}
    }    
    \caption{Scaling behavior of the baseline Jacobi method and MMRES-accelerated Jacobi iteration for (a) one-dimensional and (b) two-dimensional Poisson equations.}
    \label{fig:scaling}
\end{figure}

\section{Engineering Flow Problems}
\label{sec:cfd}

CFD in practical engineering is inherently complex, often involving intricate geometries, poor-quality meshes, stiff solution matrices, and unexpected interactions between numerical schemes. While data-driven methods for CFD have shown promising results in the literature, their adoption beyond the original developers' groups remains limited. Moreover, several recent studies have highlighted challenges in generalization, with performance degradation observed when these methods are applied outside their training datasets.  

Given the challenges, this section carefully examines the practical value of MMRES in two scenarios. The first focuses on accelerating convergence, which is straightforward and requires no further explanation. The second scenario involves obtaining a steady-state solution, which warrants clarification. For instance, in RANS simulations of flow past a square cylinder, achieving a steady-state solution is challenging due to oscillating vortex shedding. A common workaround is to impose symmetry along the centerline, effectively halving the computational domain. While this often stabilizes the solution, not all flows possess such symmetry, and methods that stabilize the solution are helpful. 
Another scenario where the stabilizing effect of the convergence booster is desirable is when an unstable steady-state solution is sought. Here, an unstable steady-state solution refers to a time-independent solution of the governing equations for which small perturbations grow over time \cite{sree2006control}. Although such solutions are not physically sustained, they serve as base flows in stability and modal analyses, providing critical insight into flow transitions, bifurcations, and the structure of phase space.
In this section, we will demonstrate that MMRES functions as a stabilizer by mitigating large-scale oscillatory modes, thereby facilitating convergence to a steady solution.

Four test cases are considered, encompassing inviscid and turbulent flows, 2D and 3D geometries.
MMRES functions as an accelerator in the first two and as a stabilizer in the last two. Details of the solver and its numerics can be found in Appendix~\ref{sec:AppendixA}. 
Here, we only note that the solver consists of pseudo time stepping and solution to linear equations inside each pseudo time step.
The latter of which is handled using point iterative method here, such as Gauss-Seidel method, and often converges in a couple of iterations.
We note that the purpose of MMRES is to reduce the number of pseudo time steps needed for a steady state solution rather than to reduce the number of Gauss-Seidel iterations needed for the linear solve. 

To reduce memory usage, we store the root-mean-square of the residuals for all flow variables instead of the residuals of each individual variable. This reduces memory requirements by 40\% in three-dimensional cases and 37.5\% in two-dimensional cases. 

In all figures, the legend follows the notation MMRES$(n_s,m)$, where $n_s$ is the number of iteration steps between two snapshots and $m$ is the total number of snapshots. Accordingly, MMRES is applied every $n_s m$ iterations. The original CFD solver without MMRES serves as the baseline for comparison.

\subsection{Subsonic Inviscid Flow over a NACA0012 Airfoil}

The performance of MMRES as an acceleration method is first evaluated for subsonic inviscid flow over a NACA0012 airfoil. The computational grid, shown in Fig.~\ref{fig:inviscid_flows_NACA0012}(a), consists of 7,038 triangular elements with 200 boundary grid points on the airfoil surface. The free-stream Mach number is set to $Ma = 0.63$, and the angle of attack is $\alpha = 2^{\circ}$. The flow is solved using a second-order finite volume method with the ROE scheme and an implicit symmetric Gauss-Seidel pseudo-time marching scheme. The computed pressure contour is shown in Fig.~\ref{fig:inviscid_flows_NACA0012}(b) for reference purposes. 

\begin{figure}[!t]
    \centering
    \subfigure {\
        \begin{minipage}[b]{.46\linewidth}
            \centering
            \begin{overpic}[scale=0.35]{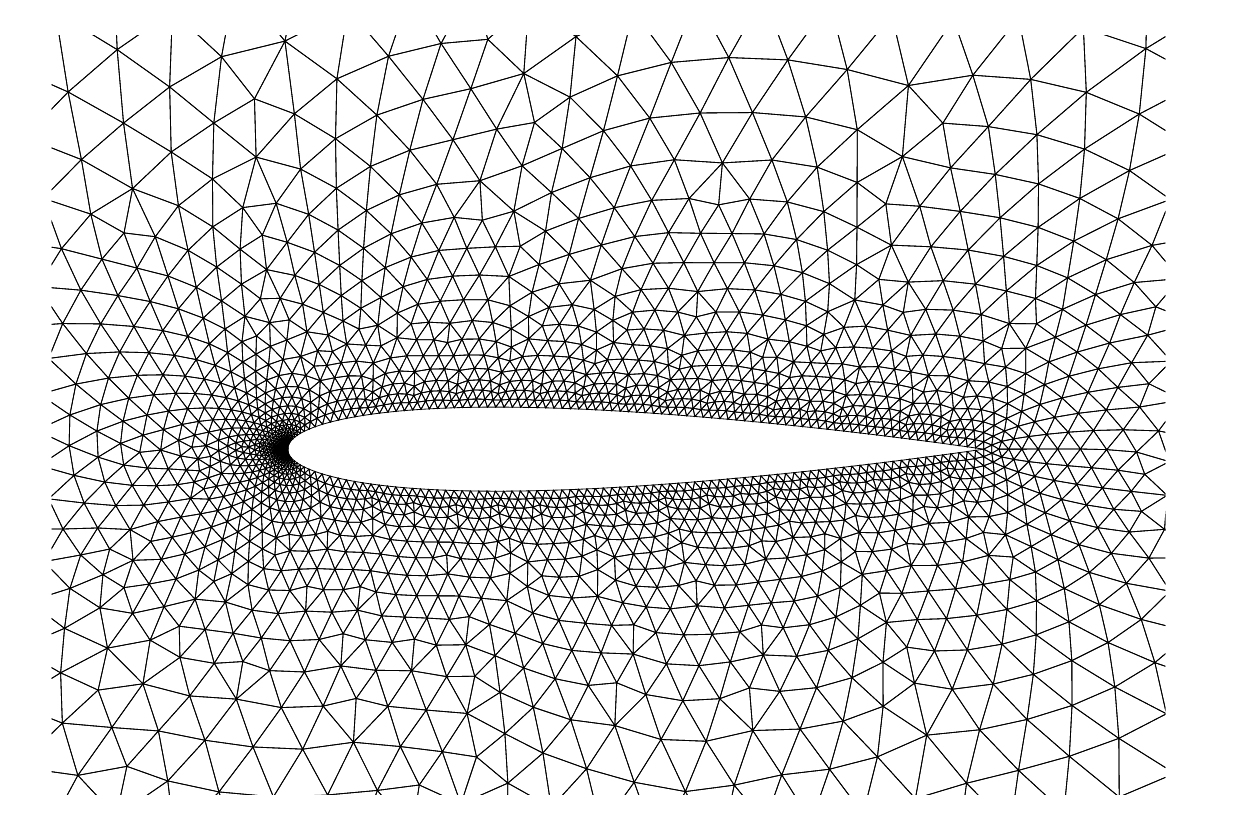}
                \put(-5,65){(a)}
            \end{overpic}
        \end{minipage}
    }
    \subfigure {\
        \begin{minipage}[b]{.46\linewidth}
            \centering
            \begin{overpic}[scale=0.35]{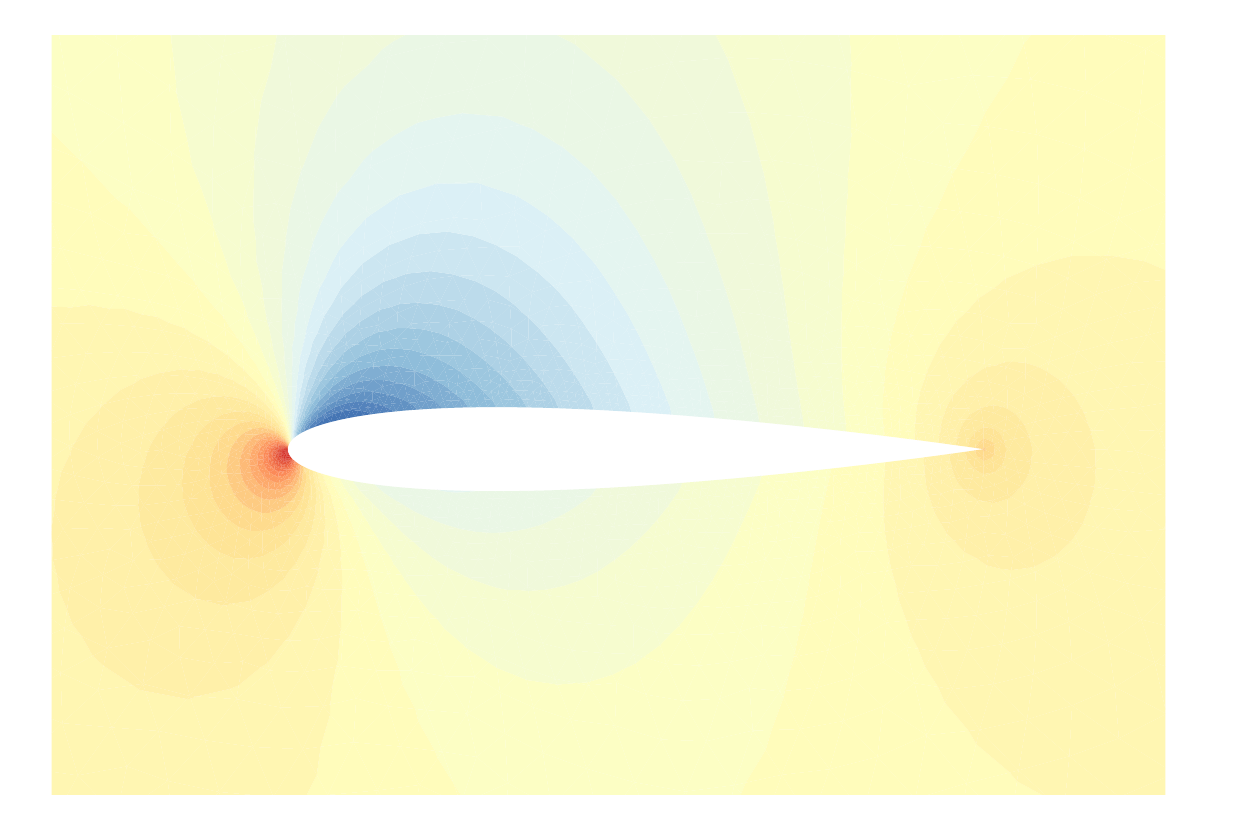}
                \put(-5,65){(b)}
            \end{overpic}
        \end{minipage}
    }    
    \caption{(a) Computational mesh and (b) computed pressure contour for subsonic inviscid flow over a NACA0012 airfoil.}
    \label{fig:inviscid_flows_NACA0012}
\end{figure}

\begin{figure}[!t]
    \centering
    \subfigure {\
        \begin{minipage}[b]{.46\linewidth}
            \centering
            \begin{overpic}[scale=0.85]{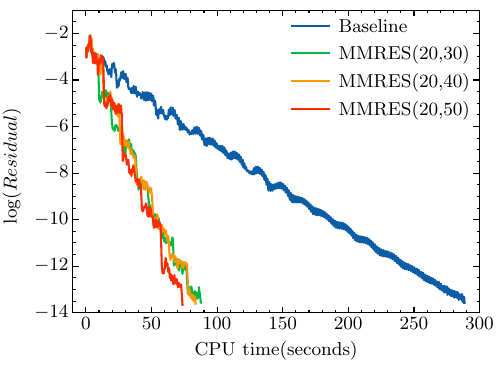}
                \put(0,75){(a)}
            \end{overpic}
        \end{minipage}
    }
    \subfigure {\
        \begin{minipage}[b]{.46\linewidth}
            \centering
            \begin{overpic}[scale=0.85]{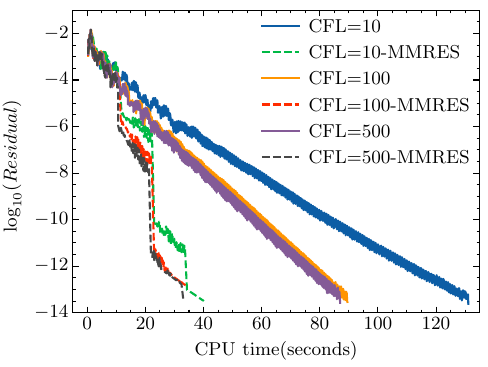}
                \put(0,75){(b)}
            \end{overpic}
        \end{minipage}
    }    
    \caption{(a) Comparison of residual convergence histories between the baseline method and MMRES-accelerated methods with different parameters;
    Here, the CFL number is kept at 2. (b) Comparison of residual convergence histories between the baseline method and MMRES-accelerated methods with different CFL numbers. Here, the MMRES method corresponds to MMRES(20,40). }
    \label{fig:resi_inviscid_NACA0012}
\end{figure}

Figure~\ref{fig:resi_inviscid_NACA0012}(a) compares the convergence histories of MMRES-accelerated simulations with different parameter settings against the baseline method. Here, the CFL number is kept at 2. The baseline solver requires 288.5 seconds to reduce the residual below $10^{-13}$, whereas MMRES-accelerated simulations achieve the same level in just 73.6 seconds, yielding a 3.9× speedup in CPU time. The additional computational cost of MMRES is negligible, averaging 0.6 seconds per acceleration step, equivalent to fewer than 40 baseline iteration steps. When MMRES is applied every 1,000 iterations as is here, the total computational overhead remains under 4\%. 

A sensitivity analysis is performed by varying the number of snapshots while keeping the snapshot interval fixed at $n_s = 20$. The total number of snapshots, $m$, is increased from 30 to 50 in increments of 10. As shown in Fig.~\ref{fig:resi_inviscid_NACA0012}(a), MMRES remains effective and insensitive to the choice of $m$. In all these cases, the CFL number is kept at 2. Further tests are conducted by increasing the CFL number to 10, 100, and 500 to evaluate MMRES performance under different time-marching conditions. In each case, MMRES is applied with 40 snapshots collected over 800 iterations. As shown in Fig.~\ref{fig:resi_inviscid_NACA0012}(b), increasing CFL improves the baseline convergence rate. However, beyond CFL = 100, further increases yield diminishing returns, suggesting a critical CFL threshold beyond which additional increases do not accelerate convergence. Across all CFL numbers, MMRES consistently achieves a 2× to 4× reduction in CPU time compared to the baseline solver. Detailed computational times for different CFL values are presented in Table~\ref{table3} for the ease of comparison.

\begin{table}
    \renewcommand\arraystretch{1.2}
    \begin{center}
        \caption{Comparison of CPU time for different CFL numbers when solving inviscid flow over a NACA0012 airfoil.}
        \begin{tabular}{c|cc}
            \toprule
            Method & CPU time (seconds) & Speedup ratio \\
            \hline
            Baseline (CFL = 2)   & 288.5 & - \\
            MMRES (CFL = 2)   & 73.6 & 3.9 \\
            \hline
            Baseline (CFL = 10)  & 133.1 & - \\
            MMRES (CFL = 10)  & 40.1 & 3.3 \\
            \hline 
            Baseline (CFL = 100)  & 89.5 & - \\
            MMRES (CFL = 100)  & 33.8 & 2.6 \\
            \hline 
            Baseline (CFL = 500)  & 87.0 & - \\
            MMRES (CFL = 500)  & 33.1 & 2.6 \\ 
            \bottomrule
        \end{tabular}
        \label{table3}
    \end{center}
\end{table}

\subsection{RANS of Transonic Flow over the ONERA M6 Wing}

Next, MMRES is applied to transonic turbulent flow over the ONERA M6 wing. The simulation conditions are Mach number $Ma = 0.8395$, Reynolds number $Re = 1.17 \times 10^7$, and angle of attack $\alpha = 3.06^{\circ}$. The computational mesh, shown in Fig.~\ref{fig:transonic_flow_ONERA_M6}(a), consists of approximately $2 \times 10^6$ elements with 33,938 surface cells. The flow is solved using a second-order finite volume method with the AUSM+ scheme, an implicit symmetric Gauss-Seidel pseudo-time marching scheme (CFL = 3), and the Spalart-Allmaras (S-A) turbulence model. The computed pressure contour is shown in Fig.~\ref{fig:transonic_flow_ONERA_M6}(b) for reference.

Two MMRES parameter sets are tested, with the total number of snapshots set to $m = 40$ and $m = 16$, and corresponding snapshot intervals of $n_s = 50$ and $n_s = 125$, respectively. The convergence histories in Fig.~\ref{fig:resi_transonic_flow_ONERA_M6}(a) show that MMRES accelerates convergence by a factor of 2.6. Additionally, the performance remains consistent between $m = 40$ and $m = 16$, confirming that MMRES is effective even with a small number of snapshots. The drag coefficient convergence history, shown in Fig.~\ref{fig:resi_transonic_flow_ONERA_M6}(b), indicates that MMRES efficiently eliminates low-frequency oscillatory error modes that would otherwise require significantly more iterations to dampen.

\begin{figure}[!t]
    \centering
    \subfigure {\
        \begin{minipage}[b]{.46\linewidth}
            \centering
            \begin{overpic}[scale=0.3]{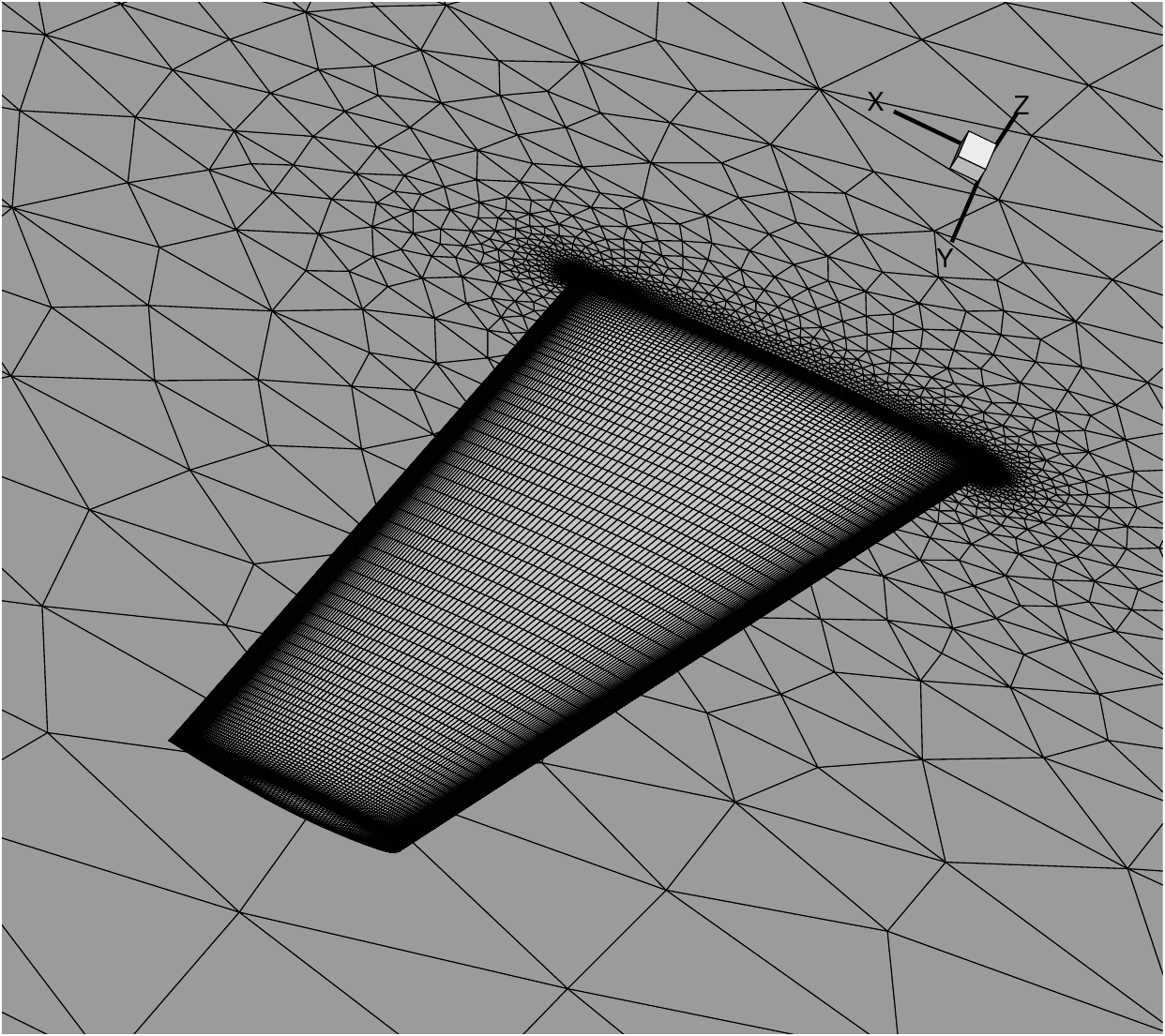}
                \put(-10,90){(a)}
            \end{overpic}
        \end{minipage}
    }
    \subfigure {\
        \begin{minipage}[b]{.46\linewidth}
            \centering
            \begin{overpic}[scale=0.3]{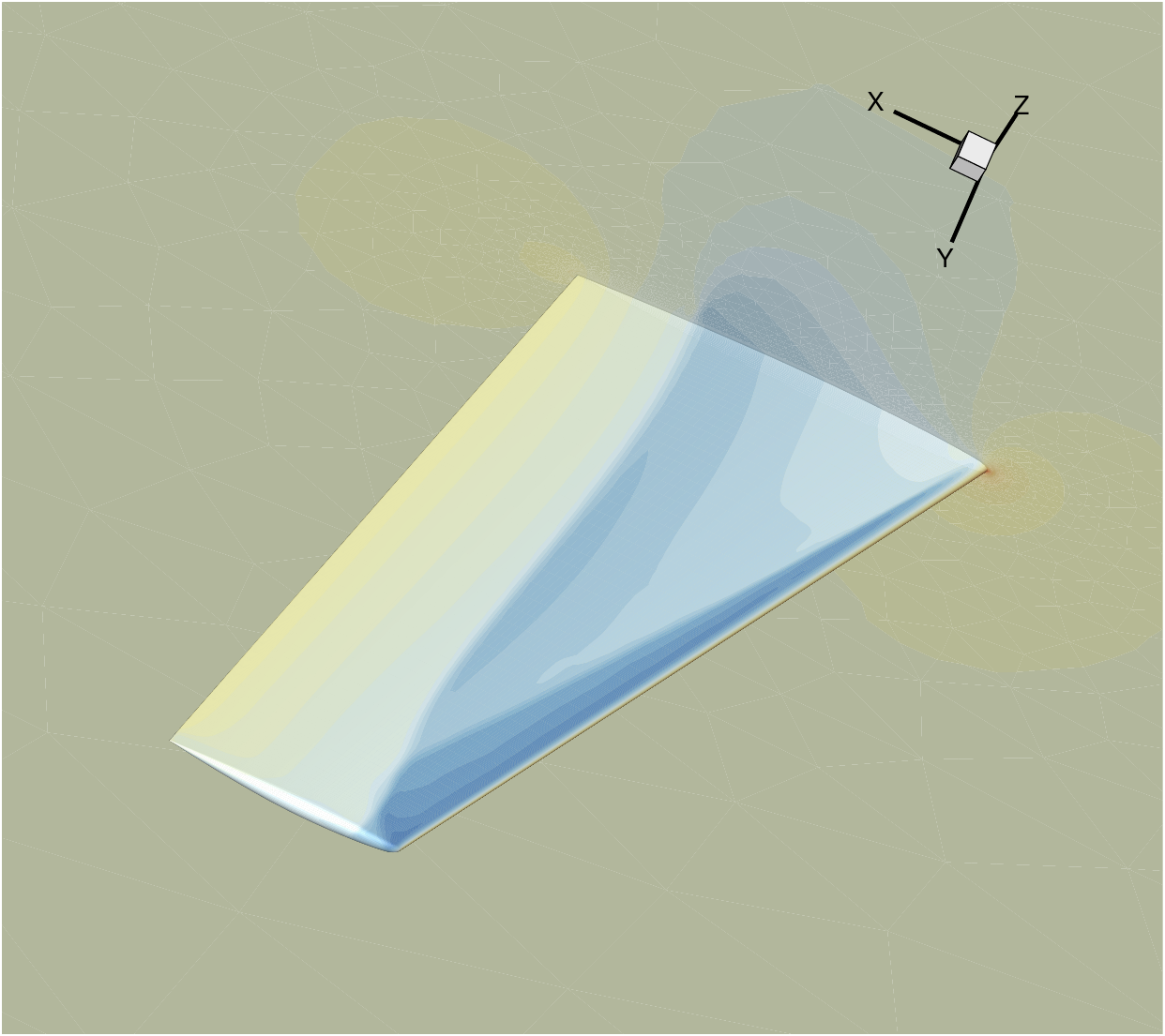}
                \put(-10,90){(b)}
            \end{overpic}
        \end{minipage}
    }    
    \caption{(a) Computational mesh; (b) computed pressure for the transonic turbulent flow over the ONERA M6 wing.}
    \label{fig:transonic_flow_ONERA_M6}
\end{figure}

\begin{figure}[!t]
    \centering
    \subfigure {\
        \begin{minipage}[b]{.46\linewidth}
            \centering
            \begin{overpic}[scale=0.85]{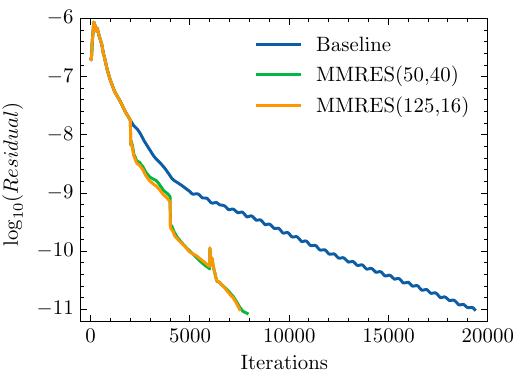}
                \put(0,70){(a)}
            \end{overpic}
        \end{minipage}
    }
    \subfigure {\
        \begin{minipage}[b]{.46\linewidth}
            \centering
            \begin{overpic}[scale=0.85]{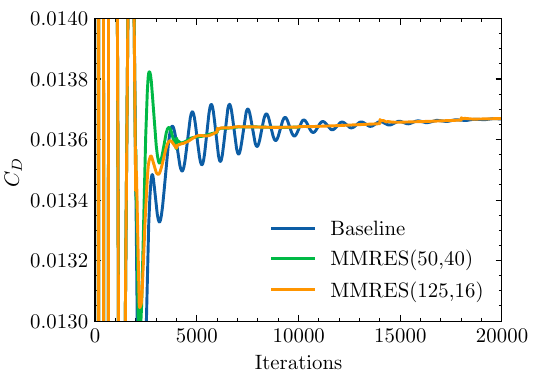}
                \put(0,70){(b)}
            \end{overpic}
        \end{minipage}
    }    
    \caption{(a) Residual convergence histories; (b) Convergence histories of drag coefficient $C_D$.}
    \label{fig:resi_transonic_flow_ONERA_M6}
\end{figure}

\subsection{Flow over a Circular Cylinder}

We now test MMRES as a stabilizer. The first case considers flow over a circular cylinder at two supercritical Reynolds numbers, $Re = 100$ and $Re = 500$. The computational domain consists of 67,613 control volumes, with 300 points along the cylinder surface. A refined mesh is used in the downstream region to accurately capture the separation zone.

We first present results for the $Re=100$ case. The objective is to obtain a fully converged solution, so we employ pseudo-time marching. The baseline method uses the LU-SGS pseudo-time marching scheme with a CFL number of 5. Figure~\ref{fig:circular_cylinder} presents the convergence histories of the residual and the lift coefficient magnitude $\vert C_l \vert$. Despite pseudo-time iteration, strong unsteady effects persist, characterized by periodic oscillations in both the residual and $\vert C_l \vert$. Furthermore, the residual stagnates, and after approximately 10,000 iterations, the flow field enters a limit cycle, characterized by periodic vortex shedding. A snapshot of this non-convergent flow is shown in Fig.~\ref{fig:circular_cylinder_Re100}(a), where the Kármán vortex street is clearly visible. 
MMRES is applied with three parameter settings: $m = 60$ and snapshot intervals of $n_s = 30$, $60$, and $90$. Figure~\ref{fig:circular_cylinder}(a) shows that all three MMRES-enhanced methods successfully reduce the residual by eight additional orders of magnitude, reaching below $10^{-14}$. In terms of $\vert C_l \vert$, MMRES reduces the lift coefficient to essentially zero within 30,000 steps, again demonstrating that the method is insensitive to the sampling window size. The computed vorticity distribution of the MMRES-enhanced method is shown in Fig.~\ref{fig:circular_cylinder_Re100}(b), with the separation zone indicated by dashed lines. The size of the separation bubble agrees well with previous studies \cite{RN35}, verifying the accuracy of the unstable steady-state solution obtained using MMRES.


\begin{figure}[!t]
    \centering
    \subfigure {\
        \begin{minipage}[b]{.46\linewidth}
            \centering
            \begin{overpic}[scale=0.85]{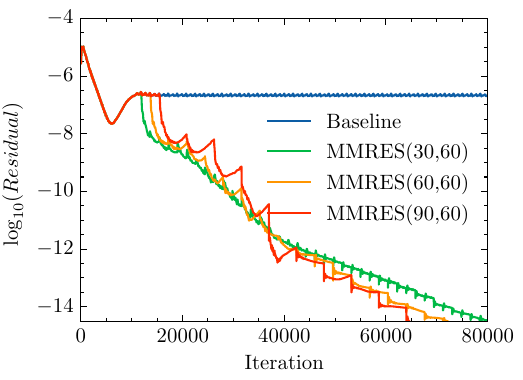}
                \put(0,70){(a)}
            \end{overpic}
        \end{minipage}
    }
    \subfigure {\
        \begin{minipage}[b]{.46\linewidth}
            \centering
            \begin{overpic}[scale=0.85]{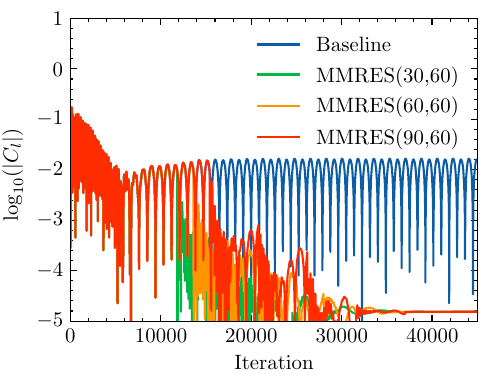}
                \put(0,75){(b)}
            \end{overpic}
        \end{minipage}
    }    
    \caption{(a) Convergence histories of the residual and (b) lift coefficient magnitude $\vert C_l \vert$ for flow over a circular cylinder at $Re = 100$.}
    \label{fig:circular_cylinder}
\end{figure}

\begin{figure}[!t]
    \centering
    \subfigure {\
        \begin{minipage}[b]{.46\linewidth}
            \centering
            \begin{overpic}[scale=0.36]{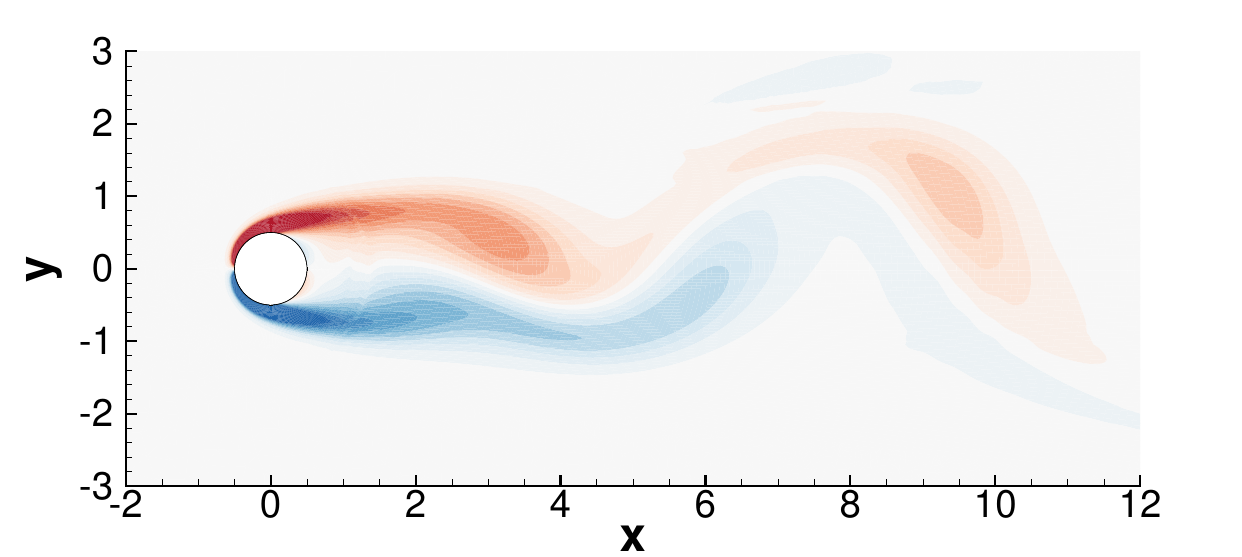}
                \put(0,45){(a)}
            \end{overpic}
        \end{minipage}
    }
    \subfigure {\
        \begin{minipage}[b]{.46\linewidth}
            \centering
            \begin{overpic}[scale=0.36]{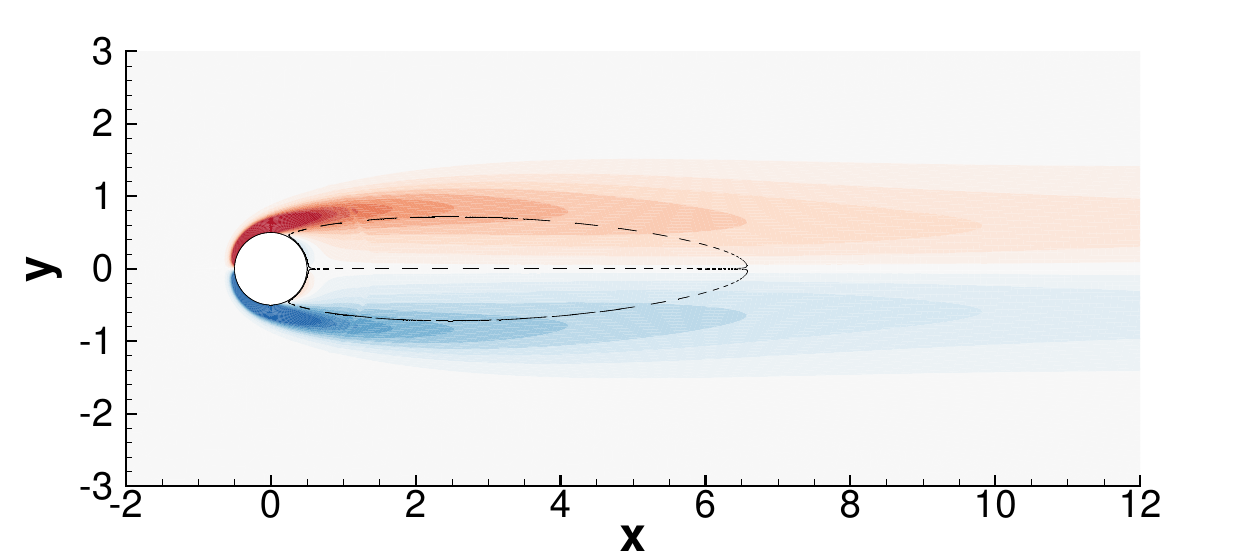}
                \put(0,45){(b)}
            \end{overpic}
        \end{minipage}
    }    
    \caption{Vorticity $(\omega = \partial_{x_1}v_2 - \partial_{x_2}v_1)$ of the flow past a circular cylinder at ${\rm Re} = 100$. (a) Snapshot of the non-convergent, unsteady flow; (b) Steady solution computed with the aid of MMRES. Dashed lines indicate separating streamlines. 
    }
    \label{fig:circular_cylinder_Re100}
\end{figure}

\begin{figure}[!t]
    \centering
    \subfigure {\
        \begin{minipage}[b]{.46\linewidth}
            \centering
            \begin{overpic}[scale=0.36]{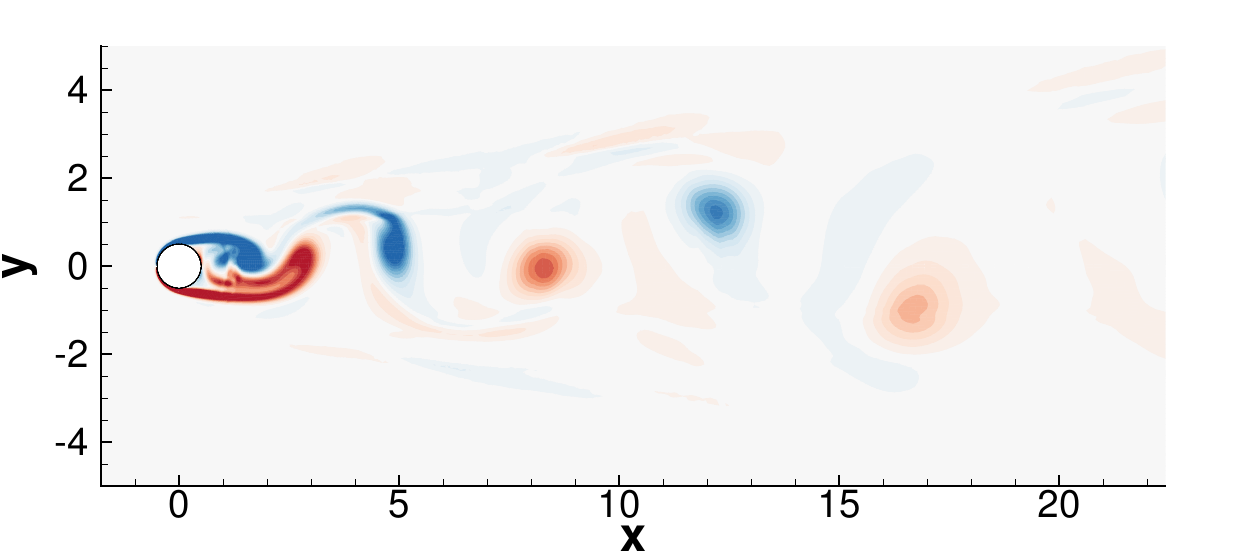}
                \put(0,45){(a)}
            \end{overpic}
        \end{minipage}
    }
    \subfigure {\
        \begin{minipage}[b]{.46\linewidth}
            \centering
            \begin{overpic}[scale=0.36]{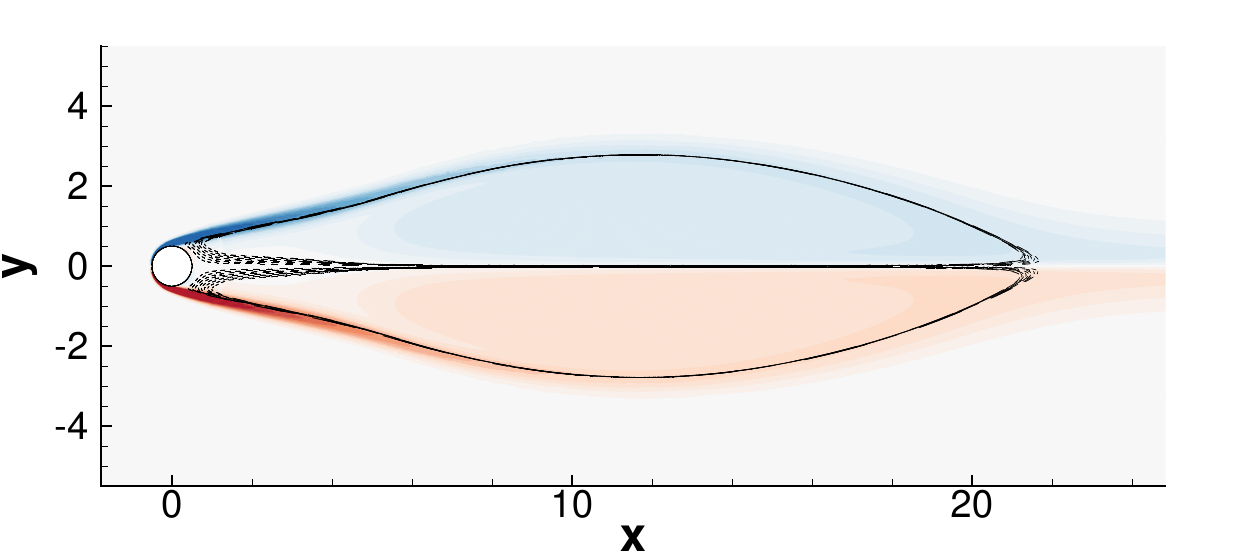}
                \put(0,45){(b)}
            \end{overpic}
        \end{minipage}
    }    
    \caption{Vorticity of the flow past a circular cylinder at ${\rm Re} = 500$. (a) Snapshot of non-convergent, unsteady flow; (b) Steady solution computed with the aid of MMRES. Dashed lines indicate separating streamlines.}
    \label{fig:circular_cylinder_Re500}
\end{figure}

The $Re=500$ case presents a greater challenge in obtaining a steady-state solution \cite{RN178}. 
For this case, the MMRES parameters are set to $m = 70$ and $n_s = 40$, corresponding to a sampling window of 2,800 steps—approximately twice the oscillation period at $Re = 500$. The computed vorticity distributions of both the non-convergent and convergent flows are shown in Fig.~\ref{fig:circular_cylinder_Re500}(a) and (b), respectively. In the MMRES-stabilized solution, the separation bubbles—plotted as dashed lines—are significantly stretched, aligning well with the reference data \cite{RN640}. This again confirms that MMRES effectively stabilizes the solution and accurately captures the steady-state flow characteristics.

\subsubsection{Transonic buffet}

While a symmetric boundary condition can still be used to obtain a steady-state solution for flow past a circular cylinder, this approach is not applicable to more complex flows, such as transonic buffet—a phenomenon of aerodynamic instability that occurs at specific combinations of Mach number and angle of attack. Transonic buffet is characterized by periodic low-frequency shock oscillations, leading to fluctuations in lift and drag. 
Obtaining a steady-state baseline RANS solution is valuable for modal analysis of the flow and subsequent analysis of the flow dynamics.
Here, transonic buffet over a NACA0012 airfoil is simulated at Mach number $Ma = 0.3$, angle of attack $\alpha = 5.5^{\circ}$, and Reynolds number $Re = 3 \times 10^6$. The computational domain consists of 18,520 elements, and the S-A turbulence model is employed for consistency with prior studies \cite{RN954}. The implicit symmetric Gauss-Seidel scheme is used with CFL = 2.

MMRES is applied with parameters $n_s = 50$ and $m = 40$, but only after the periodic limit cycle is fully developed, around the 4,000th iteration step. Figure~\ref{fig:transonic_buffet_NACA0012} compares the residual and lift coefficient $C_l$ histories. The baseline method fails to converge to a steady-state solution, while MMRES successfully suppresses oscillations and stabilizes the flow in approximately 40,000 iterations. 
For reference purposes, Fig.~\ref{fig:transonic buffet flow over NACA0012 airfoil_2}(a) presents the computed pressure contour, while Fig.~\ref{fig:transonic buffet flow over NACA0012 airfoil_2}(b) compares the pressure coefficient distribution on the airfoil surface with results obtained using a flow control method \cite{RN954}. The agreement between the two confirms the accuracy of MMRES in capturing the unstable steady-state solution. These results demonstrate that MMRES is a robust stabilizer for inherently unsteady flows.

\begin{figure}[!t]
	\centering
	\subfigure {\
		\begin{minipage}[b]{.46\linewidth}
			\centering
			\begin{overpic}[scale=0.85]{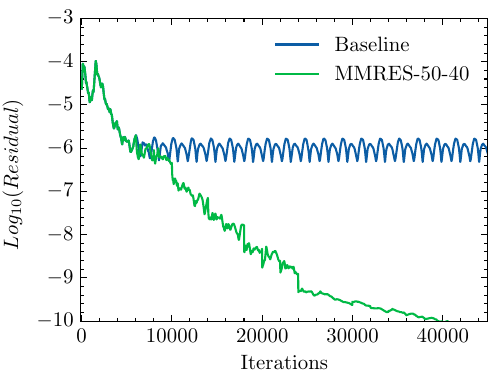}
				\put(0,70){(a)}
			\end{overpic}
		\end{minipage}
	}
	\subfigure {\
		\begin{minipage}[b]{.46\linewidth}
			\centering
			\begin{overpic}[scale=0.85]{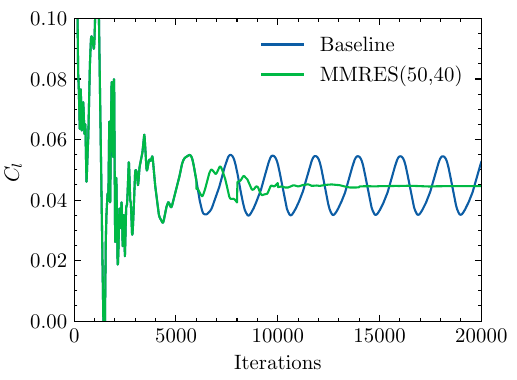}
				\put(0,70){(b)}
			\end{overpic}
		\end{minipage}
	}	
	\caption{(a) Convergence histories of residual; (b) lift coefficient $C_l$ for transonic buffet flow over NACA0012 airfoil.}
	\label{fig:transonic_buffet_NACA0012}
\end{figure}

\begin{figure}[!t]
	\centering
	\subfigure {\
		\begin{minipage}[b]{.46\linewidth}
			\centering
			\begin{overpic}[scale=0.3]{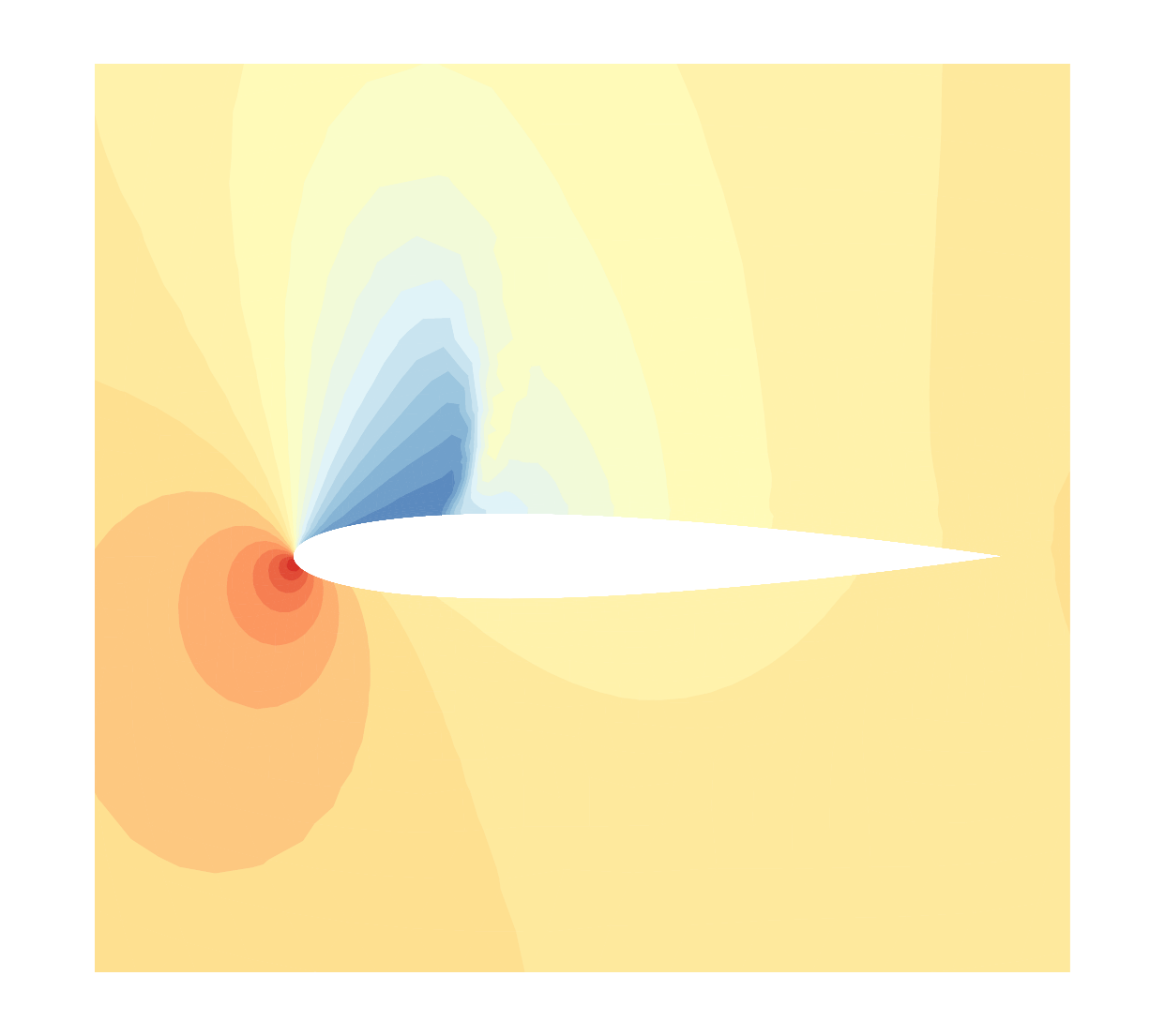}
				\put(0,80){(a)}
			\end{overpic}
		\end{minipage}
	}
	\subfigure {\
		\begin{minipage}[b]{.46\linewidth}
			\centering
			\begin{overpic}[scale=0.85]{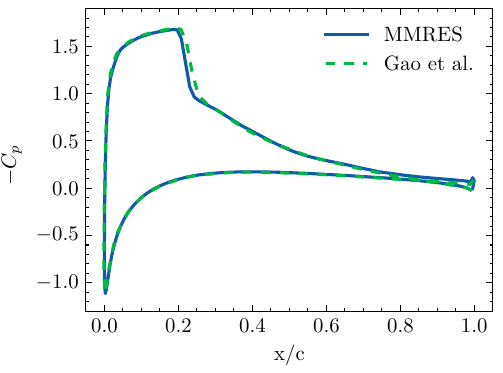}
				\put(0,70){(b)}
			\end{overpic}
		\end{minipage}
	}	
	\caption{(a) Computed pressure contour; (b) pressure coefficient distributions from the flow control method \cite{RN954} and MMRES.}
	\label{fig:transonic buffet flow over NACA0012 airfoil_2}
\end{figure}

\section{Conclusion}
\label{sec:conclusion}

This study introduces the Mean-based Minimal Residual (MMRES) method as a robust and efficient convergence booster for iterative solvers in computational fluid dynamics (CFD). By constructing a mean-based reduced-order model and minimizing the residual norm within a low-dimensional subspace, MMRES significantly accelerates convergence and stabilizes simulations that would otherwise remain unsteady. The method integrates seamlessly into existing solvers with minimal implementation effort and negligible computational overhead.

Theoretical analysis shows that MMRES reduces the time complexity of point-iterative methods from $\mathcal{O}(n^2)$ to $\mathcal{O}(n)$ for point iterative methods and linear problems. Numerical experiments on a range of realistic nonlinear CFD cases demonstrate consistent performance improvements, including $3$--$4\times$ speedup in CPU time for implicit pseudo-time stepping schemes. MMRES also enables convergence to steady-state solutions in flows with persistent unsteadiness, such as vortex shedding and transonic buffet, and facilitates the computation of unstable steady-state solutions that serve as base flows for stability and modal analyses.
Importantly, the method exhibits low sensitivity to its parameters, maintaining robust performance across a range of settings for the number of snapshots and sampling intervals. This insensitivity futher enhances its practicality and ease of use in engineering applications.

In summary, MMRES provides a practical and generalizable framework for accelerating and stabilizing CFD solvers. Future work will focus on extending the method to multi-block grids and parallel computing environments to support large-scale simulations in real-world applications.

\section{Acknowledgments}
This work was supported by the National Natural Science Fund of China (12372290) and Shaanxi Province Department of Science and Technology (2023-ZDLGY-27).

\bibliographystyle{unsrt}  


\bibliography{ref}

\section*{Appendix}
\appendix

\section{Numerical methods}
\label{sec:AppendixA}

This appendix elaborates the CFD code.
In the cartesian coordinate, the integral form of three-dimensional compressible Navier-Stokes equations can be written as:
\begin{equation}
	\label{NS_integral}
	\frac{\partial}{\partial{t}} \int_{\Omega}Ud\Omega+\int_{\partial\Omega}\left(\vec{F}(U)-\vec{G}(U,\vec{\nabla}U)\right)\cdot\vec{n}dS=0
\end{equation}
where $\Omega$ is the control volume; $\partial\Omega$ is the boundary of control volume; and $\vec{n}=[n_1,n_2,n_3]^T $ denotes the unit outward normal vector to the boundary. The $(5\times 1)$ column vector  $U$ of conservative variables, generalized $(5\times3)$ inviscid ﬂux vector $\vec{F}$ and viscous flux vector $\vec{G}$ are given as follows:
\begin{equation}
	\label{flux}
	U=\begin{bmatrix} \rho \\ \rho \vec{v} \\ \rho E \end{bmatrix},\vec{F}=\begin{bmatrix} \rho \vec{v} \\ \rho \vec{v} \bigotimes \vec{v} +p \Bar{\Bar{I}} \\ \rho \vec{v} H\end{bmatrix},\vec{G}=\begin{bmatrix} 0 \\ \Bar{\Bar{\tau}} \\ \Bar{\Bar{\tau}}\cdot\vec{v}+k\vec{\nabla}T \end{bmatrix} 	
\end{equation}
where $\rho$ is the density and $\vec{v}=[v_1,v_2,v_3]^T$ stands for the velocity vector. $E=e+{\vert\vec{v}\vert}^2/2$ is the total energy per mass and $H=E+p/\rho$ is the total enthalpy per unit mass. $\bigotimes$ and $ \vec{\nabla}$ denotes dyadic tensor and hamilton operator respectively. Finally $ \Bar{\Bar{I}}$ is the unit tensor and $\Bar{\Bar{\tau}}$ is the viscous shear stress tensor. With the assumption of linear eddy viscosity, the viscous shear stress can be written in the form:
\begin{equation}
	\label{tau}
	\tau_{ij}=(\mu+\mu_t)\left( \frac{\partial v_i}{\partial x_j}+\frac{\partial v_j}{\partial x_i}-\frac{2}{3}\frac{\partial v_k}{\partial x_k}\delta_{ij}\right) 
\end{equation}
where $\mu$ and  $\mu_t$ is the viscosity coefficient of laminar and turbulent flow. Accordingly, the thermal conductivity coefficient $ k$ can be written as:
\begin{equation}
	\label{k}
	k=\frac{c_p}{\gamma-1}\left(\frac{\mu}{\rm Pr}+\frac{\mu_t}{{\rm Pr}_t}\right) 
\end{equation}
where ${\rm Pr}$ and ${\rm Pr}_t$ is the laminar and turbulence Prandtl number respectively; and $\gamma$ is the ratio of specific heats. For the ideal gas, $\gamma$ is equal to 1.4. According to Sutherland’s law, the laminar flow viscosity coefficient is given by:
\begin{equation}
	\label{Sutherland’s law}
	\mu = \mu_{\rm ref}\frac{T_{\rm ref}+S_0}{T+S_0}\left(\frac{T}{T_{\rm ref}}\right) ^{3/2}
\end{equation}
where $T_{\rm ref}$ and $\mu_{\rm ref}$ are physical constants of reference temperature and viscosity; and $S_0$ is the Sutherland temperature. Their values are $T_{\rm ref} =273.15K, \mu_{\rm ref} = 1.716\times10^{-5}kg/(m\cdot s)$ and $S_0 =110K$, respectively. The equation of state for the ideal gas is:
\begin{equation}
	\label{equation of state}
	p=(\gamma-1)\rho\left(E-\frac{1}{2}{\vert\vec{v}\vert}^2\right).  
\end{equation}

In the cell-centered Finite Volume Method (FVM), the computational domain is divided into non-overlapping control volumes that completely cover the domain. Then the governing equation of the integral form is applied to the control volumes in the computational domain, which results to a large system of ordinary differential equations after spatial discretization. The semi-discrete formulation of the flow equation is expressed as follows:
\begin{equation}
	\label{6}
	\frac{d\bar{U}_i}{dt}=-\frac{1}{\vert\Omega_i\vert}\sum_{f\in n_{\rm f}(i)} |S_{f}|\left(\bar{F}_{f}-\bar{G}_{f}\right)
\end{equation}
where $\bar{U}_i$ denotes the cell-centered value of the control volume $ i $:
\begin{equation}
	\label{7}
	\bar{U}_i=\frac{1}{\vert\Omega_i\vert}\int_{\Omega_i}U(x_1,x_2,x_3)d\Omega.
\end{equation}
$\bar{F}_f$ and $\bar{G}_f$ denotes the face averaged normal inviscid and viscous ﬂux respectively:
\begin{equation}
	\label{face averaged}
	\bar{F}_{f}=\frac{1}{\vert S_{f}\vert}\int_{S_{f}}\vec{F}(U)dS ,\quad  \bar{G}_{f}=\frac{1}{\vert S_{f}\vert}\int_{S_{f}}\vec{G}(U,\vec{\nabla}U)dS.
\end{equation}
$\vert\Omega_i\vert$ is the volume of the control volume; $ n_{f}(i) $ is the set of face neighbor cells of cell $i$ ; $ |S_{f}| $ is the area of $f^{th}$ interface. In numerical approximation, however, they can rarely be computed exactly even if the the cell-averaged solution $\bar{U}_i$ are known. Instead, a Gaussian quadrature formula is employed to compute the face integral(take inviscid flux for example):
\begin{equation}
	\label{Gaussian quadrature}
	\bar{F}_{f}\approx\sum_{j=1}^{q} \omega_j\vec{F}\left(\vec{x}_{f,j}\right)\cdot \vec{n}_{f}
\end{equation}
where $ q $ is the total number of surface Gaussian quadrature points;  $ \omega_j $ is the corresponding weight coefficient; $\vec{n}_{f}$ is the outer normal vector of the $f^{th}$ surface of grid cell $ i $ and $\vec{x}_{f,j}$ is the position vector of Gauss quadrature point. Because variables are approximated by piece-wise polynomials, the solution is discontinuous across cell interfaces. According to the Godunov-type method, the interface normal ﬂux are calculated by the Riemann flux:
\begin{equation}
	\label{Riemann flux}
	\vec{F}\left(\vec{x}_{f,j}\right)\cdot \vec{n}_{f} \approx \hat{F}\left(U^+(\vec{x}_{f,j}),U^-(\vec{x}_{f,j}),\vec{n}_{f}\right)
\end{equation}
where  $U^+$ and $U^-$ are the solutions reconstructed inside and outside the cell $\Omega_i$, respectively. 

For the semi-discrete Eq.~\eqref{6}, to advance the solution along time, the time derivative is approximated by backwards Euler difference and the implicit discretization for flux terms are adopted in this paper. Due to the fact that the final convergent steady state is irrelevant of the physical time, the real time step $\Delta t$ is replaced by pseudo time step $\Delta \tau$: 
\begin{equation}
	\label{8}
	\frac{\bar{U}_i^{k+1}-\bar{U}_i^{k}}{\Delta \tau}+\bar{R}_i^{k+1} = 0
\end{equation}
where the subscript donates the index of pseudo time iteration step; $\bar{R}_i^{k+1} = R_i\left(\bar{U}_i^{k+1}\right)$ represents the residual of $i^{th}$ control volume and is equivalent to the right-hand side of Eq.~\eqref{6} except the minus sign.

To calculate the residual at $k+1^{th}$ pseudo time step from the current state, local linearization has to be taken based on Taylor expansion and high-order terms are neglected:
\begin{equation}
	\label{9}
	\bar{R}_i^{k+1} \approx \bar{R}_i^{k} + \sum_{j\in n_{c}(i)} \frac{\partial \bar{R}_i^k}{\partial \bar{U}_j^k}\left(\bar{U}_{j}^{k+1}-\bar{U}_{j}^{k}\right)
\end{equation}
where $n_{c}(i)$ is the set of cell $i$ and its neighbor cells (i.e. $n_{c}(i) = n_{f}(i)\cup\{i\}$).

With the definitions of 
\begin{equation}
	\left\{  
	\begin{aligned}  
		&J^{k}_{i,j}=\frac{\partial \bar{R}_i^k}{\partial \bar{U}_j^k} 
		\\  
		&J^{k} = 
		\begin{bmatrix}
			J^{k}_{1,1} & j^{k}_{1,2} & \cdots & J^{k}_{1,n}
			\\
			^{k}_{2,1} & J^{k}_{2,2} & \cdots & J^{k}_{2,n}
			\\
			\vdots & \vdots & \ddots & \vdots
			\\
			J^{k}_{n,1} & J^{k}_{n,2} & \cdots & J^{k}_{n,n}
		\end{bmatrix}   
	\end{aligned}  
	\right.
\end{equation}
and 
\begin{equation}
	\bar{U}^{k+1}=\bar{U}^{k} + \Delta \bar{U}^k 
\end{equation}
$\bar{R}_i^{k+1}$ in Eq.~\eqref{8} is substituted by the linearization term in Eq.~\eqref{9}, a implicit high-dimensional linear system will be obtained as follows:
\begin{equation}
	\label{implicit system}
	\left[\frac{I}{\Delta \tau}+J^{k}\right]\Delta \bar{U}^k = -\bar{R}^k
\end{equation}
where $ J^k \in\mathbb{R}^{dn \times dn} $ ($d=4$ for two-dimensional case and $d=5$ for three-dimensional case) is the Jacobian Matrix of the system at $ k$-th pseudo time step, which is generally extremely high-dimensional, sparse and non-sysmmetric in case of unstructured grid. $n$ is the total number of grid elements (control volumes) in FVM. 

Considering that the linear system of equations Eq.~\eqref{implicit system} is unnecessary to be solve exactly, it can be solved approximately by inner iteration algorithm(e.g. Implicit SGS, LU-SGS and GMRES) to get  the increment of iteration $\Delta\bar{U}^k$. Taking the total effects of inexact solution of Eq.~\eqref{implicit system} and classical acceleration techniques into account (local time step and residual smoothing in our solver), the iteration (pseudo time marching) can be written in a compact form if we denote $\bar{U}^k = x_k$ and $\bar{R}^k = r_k$:
\begin{equation}
	x_{k+1} = x_k + B_kr_k
\end{equation}
where $B_k \in \mathbb{R}^{n\times n}$ represents the effect of chosen iteration scheme.

\section{Remarks on MMRES}
\label{sec:AppendixB}
In this section, three remarks on MMRES are given, which compares MMRES with RRE, AA and quasi-Newton iteration:
\begin{enumerate}[label={}, leftmargin=0pt]
    \item \textbf{Remark 1: MMRES minimizes the 2-norm of residual.} 
    
    MMRES solves the LS problem in Eq.~\eqref{eq:LS}, which is equivalent to solving:
    \begin{align*}
        \Psi^{\text{T}}\Psi \xi = -\Psi^{\text{T}}\bar{r},
    \end{align*}
    leading to the solution in Eq.(~\ref{eq:optimal soultion}). This is distinguished from the problem solved in \cite{RN20}:
    \begin{align*}
        \Phi^{\text{T}}\Psi \xi = -\Phi^{\text{T}}\bar{r},
    \end{align*}
    which minimizes the $J^{\ast}$-norm of error.
    \item \textbf{Remark 2: MMRES has a close relation to Reduced Rank Extrapolation(RRE) and Anderson Acceleration(AA).} 

    MMRES has nearly the same structure as RRE and AA except that RRE and AA solve the problem:
    \begin{align*}
        {\Psi^{\prime}}^{\text{T}}{\Psi^{\prime}} \xi = -{\Psi^{\prime}}^{\text{T}}\Delta x_1
    \end{align*}
    where ${\Psi^{\prime}} = [\Delta x_2-\Delta x_1, \Delta x_3-\Delta x_2,\dots,\Delta x_m-\Delta x_{m-1}]$, which minimizes the difference between two iterations $\Delta x_i = x_{i+1}-x_i$.
    \item \textbf{Remark 3: MMRES can be interpreted as a quasi-Newton method.}
    
    The quasi-Newton method can be written as:
    \begin{equation*}
        x_{(k+1)} = x_{(k)} - J(x_{(k)})^{-1}r(x_{(k)})
    \end{equation*}
    where $J(x_{(k)})^{-1}$ is the approximated Jacobian matrix. The expression of MMRES (Eq.~\eqref{eq:optimal soultion}) can be recovered by replacing 
    $x_{(k+1)}$ and $x_{(k)}$ by $\tilde{x}^{\ast}$ and $\bar{x}$ respectively, where the inverse of Jacobian matrix $J(x_{(k)})^{-1}$ is approximated by $\Phi (\Psi^{\rm T} \Psi)^{-1} \Psi^{\rm T}$.
    
\end{enumerate}

More detailed statements about the remarks are given in the following.

\subsection{Remark 1}

Because of the linearity assumption between solution and residual, the proposed method can be formulated under the framework of general projection methods \cite{RN103}. Therefore, we can restate the problem as: 
\begin{equation}
	\label{projection problem}
	{\rm Find}\,  \tilde{x} \in \bar{x}+\mathcal{K}, \,{\rm such \, that} \, r(\tilde{x}) \perp \mathcal{L}
\end{equation}
where $\mathcal{K}$ is the search subspace and $\mathcal{L}$ is the subspace of constraints. Because  $\mathcal{K}$ is spanned by columns of $\Phi$, what we actually want to find is the optimal coefficient $\xi^{\ast}$. The constraint $r(\tilde{x}) \perp \mathcal{L}$ leads to the same LS problem in Eq.~\eqref{eq:LSprob}.

According to the definition in \cite{RN103}, projection method can be classified into \emph{orthogonal} and \emph{oblique} in case of $\mathcal{K} = \mathcal{L}$ and $\mathcal{K} \neq \mathcal{L}$ respectively. In MMRES, it is the special case that $\mathcal{L} = J^{\ast}\mathcal{K}$, which leads to a optimal result that the calculated $\tilde{x}^{\ast}$ minimize the the 2-norm of residual, as stated in Proposition 5.3 of \cite{RN103}. However, it's worth mention that although MMRES falls into the \emph{oblique} projection method, geometrically speaking, MMRES projects the residual of mean solution $\bar{r}$ \emph{orthogonally} onto the subspace $\mathcal{L}$. In other words, $-\Psi \xi^{\ast}$ is the orthogonal projection of residual of mean flow $\bar{r}$ onto $\mathcal{L}$:
\begin{equation}
	-\Psi \xi^{\ast} =\mathcal{P}_{\mathcal{L}} \bar{r}
\end{equation}
where $\mathcal{P}_{\mathcal{L}} = \Psi(\Psi^{\rm T}\Psi)^{-1}\Psi^{\rm T}$ is the orthogonal projector onto the subspace $\mathcal{L}$. As a result, the left residual $r(\tilde{x}^{\ast}) $ is equivalent to $ (I -\mathcal{P}_{\mathcal{L}})\bar{r}$ , which means $r(\tilde{x}^{\ast})$ is perpendicular to all the columns of $\Psi$(the basis matrix of subspace $\mathcal{L}$).

With the perspective of projection in hands, we can compare MMRES with a similar method proposed in \cite{RN20}. In spite of using the same mean-based ROM, the method in \cite{RN20} projected $\bar{r}$ onto $\mathcal{L}$ along the direction orthogonal to the search subspace $\mathcal{K}$. We denote this projection operators as $\mathcal{P}_{\mathcal{L}}^{\mathcal{K}}$. The properties of these two types of projection are briefly summarized in Table \ref{table1}, where the projector, resulted optimal coefficient, left residual and sense of optimality are compared in detail. 

\begin{table}
	\renewcommand\arraystretch{1.2}
	\begin{center}
		\caption{Comparison of two methods of different projection direction}
		\label{table1}
		\begin{tabular}{c|cc}
			\toprule
			Method & \qquad\qquad MMRES \qquad\qquad &  method in \cite{RN20}\\
			\hline
			Projector & $\mathcal{P}_{\mathcal{L}} = \Psi(\Psi^{\rm T}\Psi)^{-1}\Psi^{\rm T}$ & $\mathcal{P}_{\mathcal{L}}^{\mathcal{K}} = \Psi(\Phi^{\rm T}\Psi)^{-1}\Phi^{\rm T}$ \\
			Coefficient $\xi^{\ast}$ & $ -(\Psi^{\rm T}\Psi)^{-1}\Psi^{\rm T}\bar{r}$ & $ -(\Phi^{\rm T}\Psi)^{-1}\Phi^{\rm T}\bar{r}$\\
			Left residual $r(\tilde{x}^{\ast})$ & $[I-\Psi(\Psi^{\rm T}\Psi)^{-1}\Psi^{\rm T}]\bar{r}$ & $[I-\Psi(\Phi^{\rm T}\Psi)^{-1}\Phi^{\rm T}]\bar{r}$\\
			Sense of optimality & minimize $\Vert r(\tilde{x}^{\ast})\Vert_2$ & minimize $\Vert e(\tilde{x}^{\ast})\Vert_{J^{\ast}}$\\ 
			\bottomrule
		\end{tabular}
	\end{center}
	
\end{table}

It's necessary to explain the term `sense of optimality' carefully. As mentioned before, orthogonal projection of $\bar{r}$ results to the optimality in sense of minimizing the 2-norm of the residual. However, oblique projection of $\bar{r}$ along the normal direction of subspace $\mathcal{K}$ will lead to the optimality in sense of minimizing the $J^{*}$-norm of the error, where $J^{*}$-norm is defined as $\Vert x \Vert_{J^{*}} = \Vert x^{\rm T}J^{\ast}x \Vert_2$. The corresponding proposition in \cite{RN103} is repeated here for clarification.
\begin{proposition}
	Assuming $A$ is symmetric positive definite (SPD) and $\mathcal{L} = A\mathcal{K} $. Then a vector $\tilde{x}^{\ast}$ is the result of projecting $\bar{r}$ onto $\mathcal{L}$ along $\mathcal{K}$ if and only if it minimizes the $A$-norm of error over $\bar{x}+\mathcal{K}$, i.e. if and only if
	\begin{equation}
		\nonumber
		\Vert e(\tilde{x}^{\ast})\Vert_{A} = \min_{\tilde{x} \in \bar{x}+\mathcal{K}}\Vert e(\tilde{x})\Vert_{A}
	\end{equation}
	where
	\begin{equation}
		\nonumber
		e(\tilde{x}) = \tilde{x} - x^{\ast}
	\end{equation}
\end{proposition}

\begin{proof}
	Considering the definition of error and Eq.~\eqref{eq:linear_residual}, residual and error have the following relation:
	\begin{equation}
		Ae(\tilde{x}) = A\tilde{x}  - Ax^{\ast} = r(\tilde{x}).
	\end{equation}
	And because $r(\tilde{x}^{\ast})$ is orthogonal to $\mathcal{K}$ and $J^{\ast} $is SPD, we have:
	\begin{equation}
		r(\tilde{x}^{\ast})^{\rm T}\phi_i =  [Ae(\tilde{x}^{\ast})]^{\rm T}\phi_i = e(\tilde{x}^{\ast})^{\rm T}A\phi_i =  0 , \forall \phi_i \in \mathcal{K} 
	\end{equation}
	which means $e(\tilde{x}^{\ast})$ is $A$-orthogonal to subspace $\mathcal{K}$ and that is exactly the geometrical meaning of $\min_{\tilde{x} \in \bar{x}+\mathcal{K}}\Vert e(\tilde{x})\Vert_{A}$.
	
\end{proof}

It should be noticed that the choice of projector $\mathcal{P}_{\mathcal{L}}^{\mathcal{K}}$ in \cite{RN20} leads to two problems. The first one is that the vector of residual must have the same dimensions as vector of solution, which will increase the load of memory(in case of MMRES-accelerated CFD, only the energy residual vectors are collected). Secondly, as shown in Section \ref{sec:methodology}, $J^{\ast}$ is the Jacobian matrix at convergent point $x^{\ast}$, which is not SPD in most cases. This problem results that the minimization of $\Vert e(\tilde{x}^{\ast})\Vert_{J^{\ast}}$ can't be ensured in theory. Considering these two points, we can come the conclusion that MMRES is better than the method proposed in \cite{RN20}.

\subsection{Remark 2}

For the comprehensive introduction of VEM and AA, which is beyond the scope of this article, please reader refer to reference \cite{RN128,RN170,RN139,RN698,RN459}. After careful derivation, we found that MMRES has close relation with VEM and AA. To analyze their connection, a general framework is necessary and Shanks sequence transformation\cite{RN155} is adopted in our paper.

As shown in \cite{RN155}, given the same sequence of $m$ vectors $[s_1, s_2, ,\dots, s_m] \in \mathbb{R}^{n \times m}$, all these acceleration methods(MPE, RRE and AA) can be written in the form of coupled topological Shanks Transformations:
\begin{equation}
	\label{eq:Shanks}
	\tilde{s} = s_1 - [\Delta s_1,\Delta s_2,\dots,\Delta s_{m-1}](Y^{\rm T}\Delta T)^{-1}Y^{\rm T}t_1
\end{equation}
where $\Delta s_i = s_{i+1} - s_{i}, i =1,2,\dots,m-1$, $Y=[y_1, y_2, \dots, y_{m-1}]\in \mathbb{R}^{n \times (m-1)}$ and $T = [t_1, t_2,\dots, t_{m-1}] \in \mathbb{R}^{n \times (m-1)}$where $t_i, i = 1,2,\dots, m-1$ is the coupled sequence. Given the same sequence $\{ s_i\} = \{x_i\}$, different $Y$ and $T$ leads to the formulation of MPE, RRE, AA and MMRES. Different cases are summarized in Table \ref{table2}, where $\Delta^2s_i = \Delta s_{i+1} - \Delta s_{i}= s_{i+1} - 2s_{i}+s_{i-1}$. As shown in Table \ref{table2}, MMRES can be rewritten in form of Eq.~\eqref{eq:Shanks} if we apply $y_i = \Delta r_i$ and $t_i =r_i$. The detailed derivation is given as follows.

\begin{table}
	\renewcommand\arraystretch{1.2}
	\begin{center}
		\caption{Different $y_i$ and $t_i$ in various methods}
		\label{table2}
		\begin{tabular}{c|cc}
			\toprule
			Method &  $y_i$ &  $t_i$ \\
			\hline
			MPE   & $\Delta x_i$ & $\Delta x_i$ \\
			RRE/AA($\beta = 0$)   & $\Delta^{2} x_i$ & $\Delta x_i$\\
			MMRES & $\Delta r_i$ &  $ r_i$ \\ 
			\bottomrule
		\end{tabular}
	\end{center}
\end{table}

Using the Shanks sequence transform Eq.~\eqref{eq:Shanks} and setting $y_i = \Delta r_i$ and $t_i =r_i$, we have:
\begin{equation}
	\label{eq:MMRES2}
	\tilde{x} = x_1 - \Delta X(\Delta R^{\rm T}\Delta R)^{-1}\Delta R^{\rm T}r_1
\end{equation}
where $\Delta X = [\Delta x_1, \Delta x_2,\dots,\Delta x_{m-1}]$ and $\Delta R = [\Delta r_1, \Delta r_2,\dots,\Delta r_{m-1}]$. As background, recall that if a square matrix $M$ is partitioned as:
\begin{equation}
	M = \left[\begin{array}{cc}
		A & B \\
		C & D \\
	\end{array}\right]
\end{equation}
where $D$ is square and invertible, then $ {\rm det}(M) = {\rm det}(D)\times {\rm det}(M/D)$, where $(M/D)$ is the Schur complement of $D$ in $M$, that is, $(M/D) = A-BD^{-1}C$. Note that $A$ can be a $1 \times 1$ matrix, as was the case above. With this tool in hands, we can define a matrix by setting $A = x_1, B = \Delta X, C = \Delta R^{\rm T}r_1$ and $D = \Delta R^{\rm T}\Delta R$:
\begin{equation}
	M  =\left[\begin{array}{c:ccc}
		x_1 & \Delta x_1 & \cdots & \Delta x_{m-1}\\
		\hdashline
		\langle \Delta r_1,r_1 \rangle & \langle \Delta r_1,\Delta r_1 \rangle & \cdots & \langle \Delta r_1,\Delta r_{m-1} \rangle\\
		\vdots & \vdots & & \vdots\\ 
		\langle \Delta r_{m-1},r_1 \rangle & \langle \Delta r_{m-1},\Delta r_1 \rangle & \cdots & \langle \Delta r_{m-1},\Delta r_{m-1}\rangle\\ 
	\end{array}\right]
\end{equation}
where $\langle \cdot,\cdot \rangle$ denotes the inner-product of two vectors. Then Eq.~\eqref{eq:MMRES2} can be written as the ratio of determinants:
\begin{equation}
	\tilde{x} = \frac{{\rm det}(M)}{{\rm det}(D)} = \frac{\left | \begin{matrix}
			x_1 & \Delta x_1 & \cdots & \Delta x_{m-1} \\
			\langle \Delta r_1,r_1 \rangle & \langle \Delta r_1,\Delta r_1 \rangle & \cdots & \langle \Delta r_1,\Delta r_{m-1} \rangle\\
			\vdots & \vdots &  & \vdots \\
			\langle \Delta r_{m-1},r_1 \rangle &  \langle \Delta r_{m-1},\Delta r_1 \rangle & \cdots & \langle \Delta r_{m-1},\Delta r_{m-1} \rangle \\
		\end{matrix} \right | }
	{\left | \begin{matrix}
			\langle \Delta r_1,\Delta r_1 \rangle & \cdots & \langle \Delta r_1,\Delta r_{m-1} \rangle\\
			\vdots &  & \vdots \\
			\langle \Delta r_{m-1},\Delta r_1 \rangle & \cdots & \langle \Delta r_{m-1},\Delta r_{m-1} \rangle \\
		\end{matrix} \right | }.
\end{equation}
Through determinant identity transformation, we can easily get:
\begin{align}
	\tilde{x} &= \frac{\left | \begin{matrix}
			x_1 & x_2 - x_1 & \cdots & x_m - x_1 \\
			\langle r_2 - r_1,r_1 \rangle & \langle r_2 - r_1,  r_2 - r_1 \rangle & \cdots & \langle r_2 - r_1,r_m - r_1 \rangle\\
			\vdots & \vdots &  & \vdots \\
			\langle  r_m - r_1,r_1 \rangle &  \langle  r_m - r_1,r_2 - r_1 \rangle & \cdots & \langle r_m - r_1,r_m - r_1 \rangle \\
		\end{matrix} \right | }
	{\left | \begin{matrix}
			\langle r_2 - r_1,r_2 - r_1 \rangle & \cdots & \langle  r_2 - r_1,r_m - r_1 \rangle\\
			\vdots &  & \vdots \\
			\langle  r_m - r_1,r_2 - r_1 \rangle & \cdots & \langle r_m - r_1, r_m - r_1 \rangle \\
		\end{matrix} \right | }
	\\
	&= \frac{\left | \begin{matrix}
			\bar{x} & \phi_1 & \cdots & \phi_m \\
			\langle \psi_1,\bar{r} \rangle & \langle \psi_1,\psi_1 \rangle & \cdots & \langle \psi_1,\psi_m \rangle\\
			\vdots & \vdots &  & \vdots \\
			\langle \psi_m,\bar{r} \rangle &  \langle \psi_m,\psi_1 \rangle & \cdots & \langle \psi_m,\psi_m \rangle \\
		\end{matrix} \right | }
	{\left | \begin{matrix}
			\langle \psi_1,\psi_1 \rangle & \cdots & \langle \psi_1,\psi_m \rangle\\
			\vdots &  & \vdots \\
			\langle \psi_m,\psi_1 \rangle & \cdots & \langle \psi_m,\psi_m \rangle \\
		\end{matrix} \right | }
	\label{eq:MMRES_d}
\end{align}
where $\phi_i = x_i - \bar{x}$ and  $\psi_i = r_i - \bar{r}$. Finally, using \emph{Schur determinantal formula}, Eq.~\eqref{eq:MMRES_d} can be written in the form just like Eq.~\eqref{eq:optimal soultion}.

It is also worth mentioning that the choice of $t_i$ and $y_i$ in RRE, AA and MMRES have the same relation $y_i = \Delta t_i$. This property leads to the result that these methods have the same effect of minimizing the corresponding $t$ in sense of 2-norm. That is, MMRES minimizes the residual $r$. Similarly, it has been shown in previous literature that RRE and AA($\beta = 0$, $\beta$ is a parameter in AA)minimize the difference between two snapshots. However, it is hard to say which method is optimal, depending on the specific problem and more research is required.

\subsection{Remark 3}
MMRES has a close relation with quasi-Newton method and we will show in this section that MMRES is actually a simplified quasi-Newton method. Considering the nonlinear system Eq.~\ref{eq:nonlinearsystem}, the standard Newton iteration is given by:
\begin{equation}
	\label{eq:newtoniteration}
	x_{(k+1)} = x_{(k)} - J(x_{(k)})^{-1}r(x_{(k)})
\end{equation}
where the number in suberscript represents the index of Newton iteration step. To distinguish with the index of snapshot used previously, we put the index into round bracket. However, the jacobian matrix $J(U_{(k)})$ is extremely high-dimensional and can't be explicitly expressed in most cases. Therefore, an approximation,  $J_{(k)}$, is used in quasi-Newton iteration. As a subset of quasi-Newton method, the general Broyden's method\cite{RN255} updates the newest Jacobian $J_{(k)}$ from the one at previous Newton iteration step $J_{(k-m)}$ and must satisfies two conditions. The first is $m$ \emph{secant conditions}:
\begin{equation}
	\label{secant conditions}
	J_{(k)}\Delta x_{(i)} = \Delta r_{(i)}, i = 1,2,\dots,m
\end{equation}
and the second condition is \emph{no-change  condition}:
\begin{equation}
	\label{no-change condition}
	J_{(k)}q = J_{(k-m)}q \quad \forall q \quad{\rm such \, that}\quad q\perp \mathcal{K}
\end{equation}
where $\mathcal{K}$ is the subspace spanned by $\Delta x_i, i =1,2,\dots,m$. These conditions lead to the update expression of Jacobian:
\begin{equation}
	J_{(k)}^{-1} = J_{(k-m)}^{-1} + [\mathscr{X}_{(k)}-J_{(k-m)}^{-1}\mathscr{R}_{(k)}](\mathscr{R}_{(k)}^{\rm T}\mathscr{R}_{(k)})^{-1}\mathscr{R}_{(k)}^{\rm T}
\end{equation}
and the iteration:
\begin{equation}
	\begin{aligned}
		x_{(k+1)} &= x_{(k)} - J_{(k)}^{-1}r_{(k)} \\
		& = x_{(k)} - J_{(k-m)}^{-1}r_{(k)} - [\mathscr{X}_{(k)}-J_{(k-m)}^{-1}\mathscr{R}_{(k)}]\xi_{(k)} \\
	\end{aligned}
\end{equation}
where $\mathscr{R}_{(k)} = [\Delta r_{(k-m)},\dots,\Delta r_{(k)}]$, $\mathscr{X}_{(k)} = [\Delta x_{(k-m)},\dots,\Delta x_{(k)}]$ and $\xi_{(k)} = (\mathscr{R}_{(k)}^{\rm T}\mathscr{R}_{(k)})^{-1}\mathscr{R}_{(k)}^{\rm T}r_{(k)}$. 
If we ignore the information from $J_{(k-m)}^{-1}$, that is, \emph{no-change condition}~\eqref{no-change condition} is not used and $J_{(k-m)}^{-1}$ is set as zero matrix, then we will get:
\begin{equation}
	\label{simplified J}
	J_{(k)}^{-1} =  \mathscr{X}_{(k)}(\mathscr{R}_{(k)}^{\rm T}\mathscr{R}_{(k)})^{-1}\mathscr{R}_{(k)}^{\rm T}
\end{equation}
and 
\begin{equation}
	x_{(k+1)} = x_{(k)} - \mathscr{X}_{(k)}(\mathscr{R}_{(k)}^{\rm T}\mathscr{R}_{(k)})^{-1}\mathscr{R}_{(k)}^{\rm T}r_{(k)}.
\end{equation}

Recalling the derivation of MMRES in Section \ref{sec:methodology}, we will find that each MMRES cycle can be considered as a type of quasi-Newton iteration where the Jacobian matrix is approximated around the mean solution $\bar{x}$ by $m$ snapshots $[x_1, x_2, \dots, x_m]$ generated by the original iteration:
\begin{equation}
	\tilde{x}^{\ast} = \bar{x} - \bar{J}^{-1}\bar{r}
\end{equation}
where $\bar{J}^{-1}$ satisfies $m$ \emph{secant conditions}:
\begin{equation}
	\Phi = \bar{J}^{-1}\Psi
\end{equation}
where $\Phi = [x_1 - \bar{x},\dots,x_m - \bar{x}]$ and $\Psi = [r_1 - \bar{r},\dots,r_m - \bar{r}]$ are the same as that defined in Section \ref{sec:methodology}. Therefore, the inversion of Jacobian at $\bar{x}$ is approximated similar to Eq.~\eqref{simplified J}:
\begin{equation}
	\bar{J}^{-1} = \Phi(\Psi^{\rm T}\Psi)^{-1}\Psi^{\rm T}
\end{equation}

Based on what mentioned above, we can come to the conclusion that MMRES is a simplified version of quasi-Newton iteration because \emph{no-change condition} Eq.~\eqref{no-change condition} is not used and the \emph{m secant conditions} Eq.~\eqref{secant conditions} are satisfied by $m$ snapshots $x_i$ from original iteration \eqref{eq:generaliteration} instead of solution $x_{(i)}$ from previous Newton iteration Eq.\eqref{eq:newtoniteration}.

\end{document}